\documentclass[a4paper]{article}

\usepackage[utf8]{inputenc}

\usepackage{fullpage}

\usepackage{subfigure}
\usepackage[sort,numbers]{natbib}

\usepackage[shortlabels]{enumitem}

\usepackage{amsmath,amssymb,latexsym,mathtools}
\usepackage{amsthm}

\usepackage{ifpdf}
\ifpdf
\usepackage[pdftex]{graphicx}
\else
\usepackage{graphicx}
\fi

\let\O=\Omega
\let\o=\omega

\let\wto=\rightharpoonup

\providecommand{\st}{\text{ such that }}
\providecommand{\aee}{\text{a.e.}}

\let\str\epsilon 

\renewcommand{\theta}{\vartheta} 
\renewcommand{\rho}{\varrho} 
\renewcommand{\phi}{\varphi}

\renewcommand{\theta}{\vartheta} 
\renewcommand{\phi}{\varphi}
\let\e\varepsilon

\providecommand{\ue}{u^\e}

\providecommand{\vexp}{v^\e_{x'}}

\providecommand{\uea}{u^\e_\alpha}
\providecommand{\uet}{u^\e_3}
\providecommand{\ut}{{u_3}} \providecommand{\vt}{{v_3}} \providecommand{\Mtwo}{{\R}^{2\times2}}
\providecommand{\MtwoSym}{\Mtwo_{\textrm{sym}}}

\providecommand{\pt}{\partial_3}

\providecommand{\pa}{\partial_\alpha}
\providecommand{\pb}{\partial_\beta}
\providecommand{\paa}{\partial_{\alpha\alpha}}
\providecommand{\pab}{\partial_{\alpha\beta}}
\providecommand{\pbb}{\partial_{\beta\beta}}

\providecommand{\vb}{v{_\beta}}
\providecommand{\ua}{u_{\alpha}}
\providecommand{\va}{v{_\alpha}}

\providecommand{\eab}{\mathrm e_{\alpha\beta}}
\providecommand{\eaa}{\mathrm e_{\alpha\alpha}}
\providecommand{\ebb}{\mathrm e_{\beta\beta}}
\providecommand{\eat}{\mathrm e_{\alpha 3}}
\providecommand{\ett}{\mathrm e_{33}}
\providecommand{\eabe}{\eab^\e}

\providecommand{\eite}{\eit^\e}
\providecommand{\eate}{\eat^\e}
\providecommand{\ette}{\ett^\e}
\providecommand{\limeab}{e_{\alpha\beta}}
\providecommand{\limett}{e_{33}}

\newcommand{\hf}{h_f}
\newcommand{\hb}{h_b}

\newcommand{\muf}{\mu_f}
\newcommand{\mub}{\mu_b}

\newcommand{\laf}{\lambda_f}
\newcommand{\lab}{\lambda_b}

\newcommand{\Oe}{\O^\e}
\newcommand{\Of}{{\O_f}}
\newcommand{\Ob}{{\O_b}}
\newcommand{\Oef}{\O_f^\e}
\newcommand{\Oeb}{\O_b^\e}

\newcommand{\Ofilmd}{\O_f}
\newcommand{\OF}{\Ofilmd}
\newcommand{\Obondd}{\O_b}
\newcommand{\OB}{\Obondd}

\newcommand{\PeOe}{\mathcal P_\e(\Oe)}
\newcommand{\PeO}{\mathcal P_\e(\e; \O)}

\let\a=\alpha
\let\b=\beta
\let\g=\gamma
\let\d=\delta

\renewcommand{\H}{H}

\newcommand{\R}{\mathbb R}

\providecommand{\Cad}{\mathcal C}
\newcommand{\Cadw}{\Cad_\dispload}
\providecommand{\Cadz}{\Cad_0}

\providecommand{\CadKL}{\Cad_\textrm{KL}}

\providecommand{\CadrelKL}{\hat{\mathcal{C}}_\textrm{KL}}

\newcommand{\AL}[1]{}

\newcommand{\test}[1]{\hat{#1}}

\newcommand{\clo}[1]{\overline{#1}}

\providecommand{\dispload}{{w}}

\newtheorem{hypothesis}{Hypothesis}

\providecommand{\inn}{\text{ in }}
\providecommand{\onn}{\text{ on }}

\providecommand{\tr}{\mathrm{tr}}

\providecommand{\ie}{i.e.~}

\providecommand{\eg}{e.g.~}

\providecommand{\st}[0]{\superscript{st}}

\providecommand{\nl}[2][]{\left\| #1 \right\|_{#2}}

\providecommand{\pt}{\partial_3}

\providecommand{\pa}{\partial_\alpha}
\providecommand{\pb}{\partial_\beta}
\providecommand{\paa}{\partial_{\alpha\alpha}}
\providecommand{\pab}{\partial_{\alpha\beta}}
\providecommand{\pbb}{\partial_{\beta\beta}}

\providecommand{\vb}{v{_\beta}}
\providecommand{\ua}{u_{\alpha}}
\providecommand{\va}{v{_\alpha}}

\providecommand{\eab}{\mathrm e_{\alpha\beta}}
\providecommand{\eaa}{\mathrm e_{\alpha\alpha}}
\providecommand{\ebb}{\mathrm e_{\beta\beta}}
\providecommand{\ett}{\mathrm e_{33}}
\providecommand{\eat}{\mathrm e_{\alpha 3}}
\providecommand{\eit}{\mathrm e_{i 3}}

\providecommand{\strload}{{\Phi}} \providecommand{\tstrload}{{\widetilde\Phi}} \providecommand{\dispload}{{w}} 

\let\wto\rightharpoonup

\newcommand{\elast}{\mathcal A}

\newcommand{\Ae}{\strload^\e}

\newcommand{\pe}{p^\e}

\newcommand{\Aaa}{\strload_{\alpha\alpha}}

\newcommand{\Aab}{\strload_{\alpha\beta}}
\newcommand{\Att}{\strload_{33}}

\newcommand{\bAaa}{\bar\strload_{\alpha\alpha}}

\newcommand{\bAab}{\bar\strload_{\alpha\beta}}
\newcommand{\bAtt}{\bar\strload_{33}}

\providecommand{\ktt}[1]{\kappa_{33}^\e(#1)}
\providecommand{\kaa}[1]{\kappa_{\alpha \alpha}^\e(#1)}
\providecommand{\kbb}[1]{\kappa_{\beta\beta}^\e(#1)}
\providecommand{\kab}[1]{{\kappa}_{\alpha\beta}^\e(#1)}
\providecommand{\kta}[1]{\kappa_{3 \alpha}^\e(#1)}
\providecommand{\ke}[1]{\kappa^\e(#1)}
\providecommand{\kev}{{\ke{v}}}
\providecommand{\kttv}{\ktt{v}}
\providecommand{\kaav}{\kaa{v}}
\providecommand{\kbbv}{\kbb{v}}
\providecommand{\kabv}{\kab{v}}
\providecommand{\ktav}{\kta{v}}
\providecommand{\kue}{   \kappa^\e(\ue)}
\providecommand{\kttue}{ \kappa_{33}^\e}
\providecommand{\kaaue}{ \kappa_{\alpha \alpha}^\e}

\providecommand{\kabue}{{\kappa}_{\alpha\beta}^\e}
\providecommand{\ktaue}{ \kappa_{3 \alpha}^\e}

\providecommand{\kebl}[1]{\hat\kappa^\e(#1)}
\providecommand{\kttbl}[1]{\hat{\kappa}_{33}^\e(#1)}
\providecommand{\kaabl}[1]{\hat{\kappa}_{\alpha \alpha}^\e(#1)}
\providecommand{\kbbbl}[1]{\hat{\kappa}_{\beta\beta}^\e(#1)}
\providecommand{\kabbl}[1]{{\hat{\kappa}}_{\alpha\beta}^\e(#1)}
\providecommand{\ktabl}[1]{\hat{\kappa}_{3 \alpha}^\e(#1)}

\providecommand{\kevbl}{{\kebl{v}}}
\providecommand{\kttblv}{\kttbl{v}}
\providecommand{\kaablv}{\kaabl{v}}
\providecommand{\kbbblv}{\kbbbl{v}}
\providecommand{\kabblv}{\kabbl{v}}
\providecommand{\ktablv}{\ktabl{v}}

\providecommand{\kblue}{\hat \kappa^\e(\ue)}
\providecommand{\kttblue}{\hat \kappa_{33}^\e(\ue)}
\providecommand{\kaablue}{\hat \kappa_{\alpha \alpha}^\e(\ue)}

\providecommand{\kabblue}{\hat \kappa_{\alpha\beta}^\e(\ue)}
\providecommand{\ktablue}{\hat \kappa_{3 \alpha}^\e(\ue)}

\providecommand{\limkbl}{\hat k}
\providecommand{\limkttbl}{\hat k_{33}}
\providecommand{\limkaabl}{\hat k_{\alpha \alpha}}

\providecommand{\limkabbl}{\hat k_{\alpha\beta}}

\providecommand{\limkf}{k}
\providecommand{\limkijf}{k_{ij}}
\providecommand{\limkttf}{k_{33}}
\providecommand{\limkaaf}{k_{\alpha \alpha}}

\providecommand{\limkabf}{{k}_{\alpha\beta}}
\providecommand{\limktaf}{k_{3 \alpha}}

\providecommand{\io}{\int_{\o}}
\providecommand{\iof}{\int_{\OF}}
\providecommand{\iob}{\int_{\OB}}

\let\O=\Omega
\let\o=\omega

\newcommand{\op}{\o_+}
\newcommand{\om}{\o_-}
\newcommand{\oem}{\o_-^\e}
\newcommand{\oep}{\o_+^\e}

\newcommand{\oz}{\o_0}

\usepackage[noblocks]{authblk}

\newtheorem{theorem}{Theorem}[section]

\newtheorem{lemma}{Lemma}[section]

\theoremstyle{remark}
\newtheorem{remark}{Remark}[section]

\usepackage{color}

\title{On the asymptotic derivation of Winkler-type energies from 3D elasticity}

\author[1]{Andr\'es A. Le\'on Baldelli\thanks{Electronic address: \texttt{leonbaldelli@maths.ox.ac.uk}; Corresponding author}}

\author[1,2]{B. Bourdin\thanks{Electronic address: \texttt{bourdin@lsu.edu};}}

\affil[1]{Center for Computation \& Technology, Louisiana State University, Baton Rouge LA 70803, USA}
\affil[2]{Department of Mathematics, Louisiana State University, Baton Rouge LA 70803, USA}

\begin{document}

\maketitle

\begin{abstract}
We show how bilateral, linear, elastic foundations (i.e. Winkler foundations) often regarded as heuristic, phenomenological models, emerge asymptotically from standard, linear, three-dimensional elasticity.
We study the parametric asymptotics of a non-homogeneous linearly elastic bi-layer attached to a rigid substrate as its thickness vanishes, for varying thickness and stiffness ratios.
By using rigorous arguments based on energy estimates, we provide a first rational and constructive justification of reduced foundation models.
We establish the variational weak convergence of the three-dimensional elasticity problem to a two-dimensional one, of either a ``membrane over in-plane elastic foundation'', or a ``plate over transverse elastic foundation''.
These two regimes are function of the only two parameters of the system, and a phase diagram synthesizes their domains of validity.
Moreover, we derive explicit formul\ae\ relating the effective coefficients of the elastic foundation to the elastic and geometric parameters of the original three-dimensional system.
\end{abstract}

\section{Introduction} \makeatletter{}
\label{sec:introduction}

We focus on models of linear, bilateral, elastic foundations, known as ``Winkler foundations'' (\cite{Winkler1867}) in the engineering community. 
Such models are commonly used to  account for the bending of beams supported by elastic soil, represented by a continuous bed of mutually independent, linear, elastic, springs. 
They involve a single parameter,  the ratio between the ``bending modulus'' of the beam and the ``equivalent stiffness'' of the elastic foundation, henceforth denoted by $k$.
As a consequence, the pressure $q(x)$ exerted by the elastic foundation at a given point in response to the vertical displacement $u(x)$ of the overlying beam, takes the simple form:
\begin{equation}
    \label{eqn:winkler_const_eq}
		q(x)=K u(x).
\end{equation}
Such type of foundations, straightforwardly extended to two dimensions, have found application in the study of 
the static and dynamic response of embedded caisson foundations~\cite{Gerolymos2006,Zafeirakos2013}, 
supported shells~\citep{Paliwal1996}, 
filled tanks~\cite{Amabili1998},
free vibrations of nanostructured plates~\cite{Pradhan2009}, 
pile bending in layered soil~\cite{Sica2013}, 
seismic response of piers~\cite{Chen2003}, 
carbon nanotubes embedded in elastic media~\cite{Pradhan2011},
chromosome function~\cite{Kleckner2004}, etc.
Analogous reduced models, labeled ``shear lag'', have been employed after the original contribution of~\cite{Cox1952a} to analyze the elastic response of matrix-fiber composites under different material and loading conditions, see~\cite{Hutchinson1990,Nairn1997,Nairn2001,Jiang2008} and references therein.

Linear elastic foundation models have also kindled the interest of the theoretical mechanics community.
Building up on these models, the nonlinear response of complex systems has been studied in the context of formation of geometrically involved wrinkling buckling modes in thin elastic films over compliant substrates~\cite{Audoly2008a,Audoly2008b,Audoly2008c}, in the analysis of fracture mechanisms in thin film systems~\cite{Xia2000}, further leading to the analysis of the emergence of quasi-periodic crack structures and other complex crack patterns, as studied in~\cite{LeonBaldelli2012,Mesgarnejad,LeonBaldelli2014} in the context of variational approach to fracture mechanics.

Winkler foundation models are regarded as heuristic, phenomenological models, and their consistency on the physical ground is often questioned in favor of more involved multi-parameter foundation models such as Pasternak~\cite{Pasternak1955}, Filonenko-Borodich~\cite{Filonenko-Borodich1940}, to name a few.
The choice of such model is usually entrusted to mechanical intuition, and  the calibration of the ``equivalent stiffness'' constant $K$ is usually performed with empirical tabulated data, or finite element computations.

Despite their wide application, to the best knowledge of the authors and up to now, no attempts have been made to fully justify and derive linear elastic foundation models from a general, three-dimensional elastic model without resorting to any a priori kinematic assumption. 

The purpose of this work is to give insight into the nature and validity of such reduced-dimension models, via a mathematically rigorous asymptotic analysis, providing a novel justification of Winkler foundation models.

As a product of the deductive analysis, we also obtain the dependence of the ``equivalent stiffness'' of the foundation, $K$ in Equation~\eqref{eqn:winkler_const_eq}, on the material and geometric parameters of the system.

\bigskip

In thin film systems, the separation of scales between in-plane and out-of-plane dimensions introduces a ``small parameter'', henceforth denoted by $\e$, that renders the variational elasticity problem an instance of a ``singular perturbation problem'' which can be tackled with techniques of rigorous asymptotic analysis, as studied in an abstract setting in~\cite{Lions1973}. 
Such asymptotic approaches have also permitted the rigorous justification of linear and nonlinear, reduced dimension, theories of homogeneous and heterogeneous~\cite{Geymonat1987a,Marigo2005a} rods as well as linear and nonlinear plates~\cite{Ciarlet1979} and shells~\cite{Ciarlet1996c}.

Engineering intuition suggests that there may be multiple scenario leading to such reduced model.
Our interest in providing a rigorous derivation span from previous works on system of thin films bonded to a rigid substrate, hence we focus on the general situation of inearly elastic bi-layer system, constituted by a \emph{film} bonded to a rigid substrate by the means of a \emph{bonding layer}.
We take into account possible abrupt variations of the elastic (stiffness) and geometric parameters (thicknesses) of the two layers by prescribing an arbitrary and general scaling law for the stiffness and thickness ratios, depending on the geometric small parameter $\e$.

The work is organized as follows.
In Section~\ref{sec:statement_of_the_problem_and_main_results}, we introduce the asymptotic, three-dimensional, elastic problem $\PeOe$ of a bi-layer system attached to a rigid substrate, in the framework of geometrically linear elasticity. 
We further state how the \emph{data}, namely the intensity of the loads, the geometric and material parameters are related to $\e$.
In order to investigate the influence of material and geometric parameters rather than the effect  of the order of magnitude of the imposed loads on the limitig model, as \eg in the spirit of \cite{Marigo2005a},
we prescribe a fixed scaling law for the load and a general scaling law for the material and geometric quantities (thicknesses and stiffnesses), both depending upon a small parameter $\e$. 
The latter identifies an $\e$-indexed family of energies $\tilde E_\e$ whose associated minimization problems we shall study in the limit as $\e\to 0$.
We then perform the classical anisotropic rescaling of the space variables, in order to obtain a new problem $\PeO$, equivalent to $\PeOe$, but posed on a fixed domain $\O$ and whose dependence upon $\e$ is explicit.
We finally synthetically illustrate on a phase diagram identified by the two non-dimensional parameters of the problem, the various asymptotic regimes reached in the limit as $\e\to0$.

In Section~\ref{sec:proofs_of_the_theorems} we establish the main results of the paper by performing the parametric asymptotic analysis of the elasticity problems of the three-dimensional bi-layer systems. 
We start by establishing a crucial lemma, namely Lemma~\ref{lem:conv_scaled_strains}, which gives the convergence properties of the families of \emph{scaled strains}. 
We finally move to the proof of the results collected into Theorem~\ref{thm:in_of_plane_elast_found} and~\ref{thm:out_of_plane_elast_found}.
The analysis of each regime is concluded by a dimensional analysis aimed to outline the distinctive feature of such reduced models, namely the existence of a characteristic \emph{elastic} length scale in the limit equations.

\makeatletter{}\section{Statement of the problem and main results} \label{sec:statement_of_the_problem_and_main_results}
\subsection{Notation} \label{sub:notation}
We denote by $\O$ the reference configuration of a three-dimensional linearly elastic body and by $u$ its displacement field.
We use the usual notation for function spaces, denoting by $L^2(\O;\R^n), H^1(\O;\R^n)$  respectively the Lebesgue space of square integrable functions on $\O$ with values in $\R^n$, the Sobolev space of square integrable functions with values in $\R^n$ with square integrable weak derivatives on $\O$.
We shall denote by $H^1_0(\O;\R^n)$ the vector space associated to $H^1(\O, \R^n)$, and use the concise notation $L^2(\O), H^1(\O), H^1_0(\O)$ whenever $n=1$.
The norm of a function $u$ in the normed space $X$ is denoted by $\nl[u]{X}$, whenever $X=L^2(\O)$ we shall use the concise notation $\nl[u]{\O}$. Lastly, we denote by $\dot H^1(\O)$ the quotient space between $H^1(\O)$ and the space of infinitesimal rigid displacements $\mathcal R(\O) = \{v\in H^1(\O), e_{ij}(v)=0\}$, equipped by its norm 
$\nl[u]{\dot H^1(\O)} := \inf_{r\in \mathcal R(\O)}\nl[u-r]{H^1(\O)}$.
Weak and strong convergences are denoted by $\wto$ and $\to$, respectively.

We shall denote by $\CadKL(\O)=\left\{v \in \H^1(\O; \R^3), \eit(v)=0 \inn \O \right\}$ the space of sufficiently smooth shear-free displacements in $\O$, and by
$\CadrelKL(\O):=\left\{ \dot\H^1(\O)\cap \mathcal R(\O)^\perp \times H^1(\O), \mathrm{e}_{i3}(v)=0 \inn \OF \right\}$ the admissible space of sufficiently smooth displacements whose in-plane components are orthogonal to infinitesimal rigid displacements, whose transverse component satisfies the homogeneous Dirichlet boundary condition on the interface $\om$, and which are shear-free in the film. 
Classically, $\e\ll 1$ is a small parameter (which we shall let to $0$), and the dependence of functions, domains and operators upon $\e$ is expressed by a superscripted $\e$. 
Consequently, $x^\e$ is a material point belonging to the $\e$-indexed family of domains $\Oe$.
We denote by $e^\e(v)$ the linearized gradient of deformation tensor of the displacement field $v$, defined as $e^\e(v)=1/2(\nabla^\e v + (\nabla^\e)^T v)=1/2 \left( \frac{\partial v_i}{\partial x^\e_j} + \frac{\partial v_j}{\partial x^\e_i} \right)  $ 
In all that follows, subscripts $b$ and $f$ refer to quantities relative to the bonding layer and film, respectively.
The inner (scalar) product between tensors is denoted by a column sign, their components are indicated by subscripted roman and greek letters spanning the sets $\{1, 2, 3\}$ and $\{1, 2\}$, respectively.

\bigskip

We consider as model system consisting of two superposed linearly elastic, isotropic, piecewise homogeneous layers bonded to a rigid substrate, as sketched in Figure~\ref{fig:fig_3d}.
Let $\o$ be a bounded domain in $\R^2$ of characteristic diameter $ L=\operatorname{diam}(\o)$.
 A  \emph{thin} film occupies the region of space $\clo{\Oef} = \clo{\o} \times [0, \e\hf]$ with $\e\ll 1$, and the bonding layer occupies the set $\clo{\Oeb} = \clo{\o}\times[-\e^{\a+1}\hb, 0]$ for some constant $\a\in \R$. 
The latter is attached to a rigid substrate which imposes a Dirichlet (clamping) boundary condition of place 
 at the interface $\oem:=\o\times\{-\e^{\a+1}\hb\}$, with datum $w\in L^2(\o)$. 
 We denote the entire domain by $\Oe:=\Oef\cup\Oeb$.

\begin{figure}[htbp]
 \centering
    \includegraphics[width=9cm]{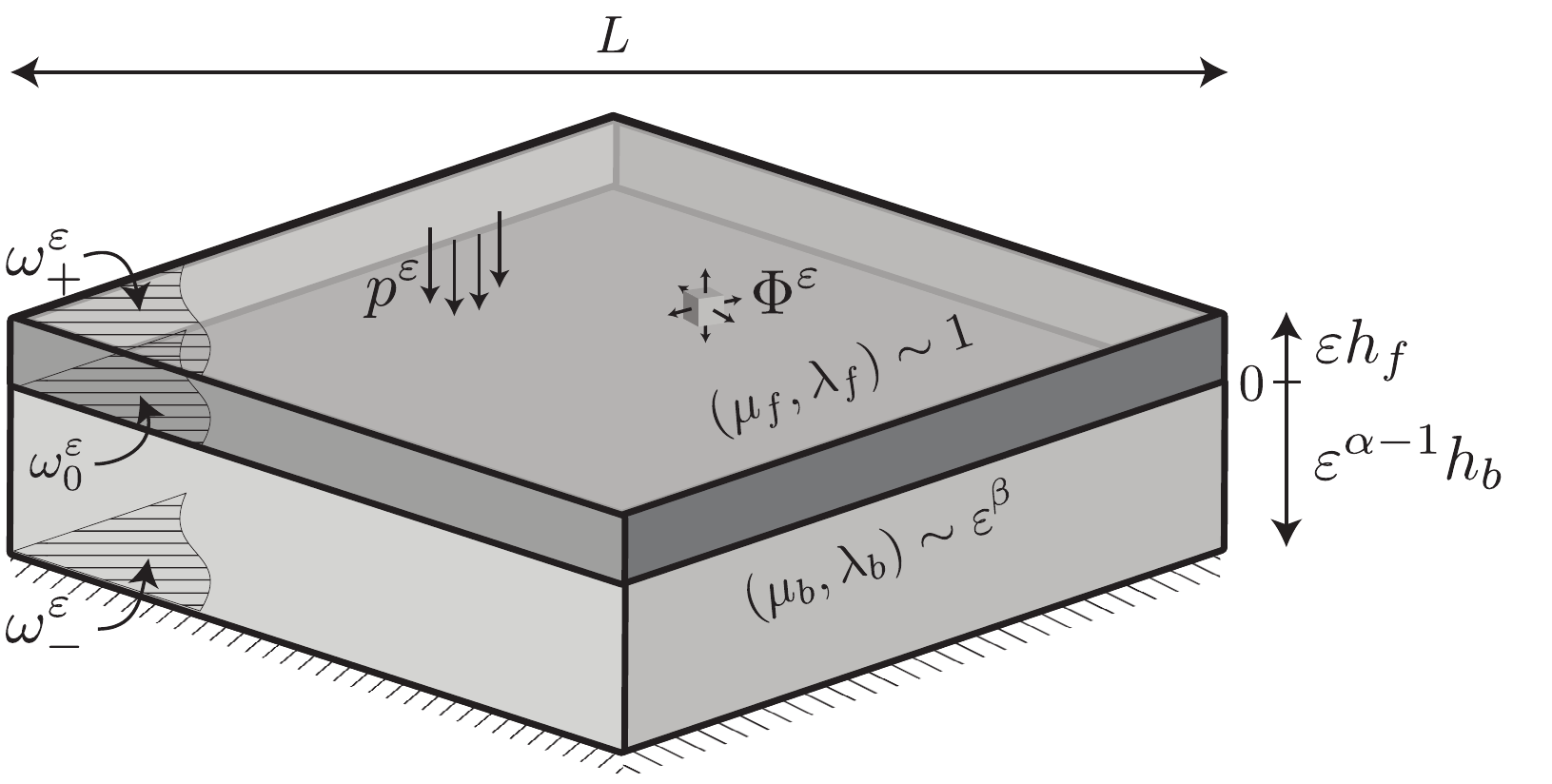}
 \caption{The three dimensional model system.}
 \label{fig:fig_3d}
\end{figure}
Considering the substrate infinitely stiff with respect to the overlying film system, the boundary datum $w$ is interpreted as the displacement that the underlying substrate would undergo under structural loads, neglecting the presence of the overlying film system.
In addition to the hard load $w$, we consider two additional loading modes: an imposed inelastic strain $\tstrload^\e\in L^2(\Oe; \R^{3\times3})$ and a transverse force $\pe\in L^2(\op)$ acting on the upper surface.
The inelastic strain can physically be originated by, e.g., temperature change, humidity or other multiphysical couplings, and is typically the source of in-plane deformations.
On the other hand, transverse surface forces may induce bending.
Taking into account both in-plane and out-of-plane deformation modes, we model both loads as independent parameters regardless of their physical origin. 
Finally, the lateral boundary $\partial \o \times (-\e^{\a+1}\hb, \hf)$
  is left free.

The Hooke law for a linear elastic material writes $\sigma^\e=\elast^\e(x) \str=\lambda^\e(x)\tr(\str){I}_3+2\mu^\e(x) \str$. 
Here, $\str$ stands for the linearized elastic strain and $\elast^\e(x)$ is the fourth order stiffness tensor.
Classically, the potential elastic energy density $W(\str^\e(v); x)$ associated to an admissible displacement field $v$, is a quadratic function of the elastic strain tensor $\str^\e(v)$ and reads:
\[
	W^\e(\xi; x)=\elast^\e(x) \xi: \xi = \lambda^\e(x) \tr(\xi)^2 + 2 \mu^\e(x) \xi:\xi,
\]
where the linearized elastic strain tensor $\str^\e(v, x)=e^\e(v)-\tilde \strload^\e(x)$ accounts for the presence of imposed inelastic strains $\tilde \strload^\e(x)$. 
Denoting by $\mathcal L^\e(u)=\int_{\oep} \pe v_3 ds$ the work of the surface force, the total potential energy $\tilde E(v)$ of the bi-layer system subject to inelastic strains and transverse surface loads reads:
\begin{equation}
\label{eqn:total_en_3d}
	\tilde E_\e(v):=\frac{1}{2}\int_{\Oe} W^\e( \str^\e (v, x), x)-\mathcal L^\e(v)
					\end{equation}
and is defined on kinematically admissible displacements belonging to the set $\Cadw^\e$ of sufficiently smooth, vector-valued fields $v$ defined on $\Oe$ and satisfying the condition of place $v=\dispload$ on $\oem$, namely:
\[
	\Cadw^\e(\O):=\left\{ v_i \in \H^1(\Oe), v_i=w \onn \oem \right\}.
\]	

Up to a change of variable, we can bring the imposed boundary displacement into the bulk;
in addition, without restricting the generality of our arguments and in order to keep the analysis as simple as possible, we further consider inelastic strains of the form:
\[
	\tilde \strload^\e(x) = 
	\begin{cases}
		\strload^\e(x), &\text{ if }x \in \Of\\
		0, &\text{ if }x \in \Ob\\
	\end{cases},
\]

For the definiteness of the elastic energy \eqref{eqn:total_en_3d}, we have to specify how the data, namely (the order of magnitude of) the material coefficients in $\elast^\e(x)$ as well as the intensity of the loads $\strload^\e $ and $ \pe$, depend on $\e$.
As far as the dimension-reduction result is concerned, multiple choices are viable, possibly leading to different limit models.
Our goal is to highlight the key elastic coupling mechanisms arising in elastic multilayer structures, with particular focus on the influence of the material and geometric parameters on the limit behavior, as opposed to analyze the different asymptotic models arising as the load intensity (ratio) changes, as done \eg in \cite{Marigo2005a,Friesecke2006}. 
We shall hence account for a wide range of relative thickness ratios and for possible strong mismatch in the elasticity coefficients, considering the simplest scaling laws that allow us to explore the elastic couplings yielding linear elastic foundations as an asymptotic result.
Hence, we perform a parametric study, letting material and geometric parameters \emph{vary}, for a \emph{fixed} a scaling law for the intensities of the external loads.
More specifically, we assume the following hypotheses.
\begin{hypothesis}[Scaling of the external load]
\label{hyp:geom_elast_scaling_law}
Given functions $p\in L^2(\o), \strload\in L^2(\O; \Mtwo)$, we assume that the magnitude of the external loads scale as:
  \begin{equation}
\label{eqn:scaling_law_loads} 
\pe(x)=\e^2 p(x), \qquad \Ae(x) = \e \strload(x)
 \end{equation}
with $\strload\in L^2(\Of)$.
\end{hypothesis}

\begin{remark}
Owing to the linearity of the problem, up to a suitable rescaling of the unknown displacement and of the energy, the elasticity problem is identical under a more general scaling law for the loads of the type: $\pe = \e^{t+1} p, \Ae=\e^{t}$ for $t\in \R$.
Indeed, only the relative order of magnitude of the elastic load potentials associated to the two loading modes is relevant. Hence, without any further loss of generality, we take $t=1$.
\end{remark}

\begin{hypothesis}[Scaling of material properties]
\label{hyp:elast_scaling_law}
Given a constant $\beta\in \R$, we assume that 
 the elastic moduli of the layers scale as:
   \begin{equation}
     \frac{E_b^\e}{E^\e_f}= \rho_E \e^\b, \quad\frac{\nu_b^\e}{\nu_f^\e}= \rho_\nu, 
 \label{eqn:scaling_law_elast}
   \end{equation}
  where $\rho_E$ and $\rho_\nu$ 
   are non-dimensional coefficients independent of $\e$.
 \end{hypothesis}

\begin{remark}
Note that this is equivalent to say that both film to bonding layer ratios of the Lamé parameters scale as $\e^\b$ and no strong elastic anisotropy is present so that the scaling law \eqref{eqn:scaling_law_elast} is of the form:
\[
	 \dfrac{\mu_b}{\mu_f}= \rho_\mu \e^\b, \quad\dfrac{\lambda_b^\e}{\lambda_f^\e}= \rho_\lambda \e^\b, 
\]
where $\rho_\mu, \rho_\lambda\in \R$ are independent of $\e$.
Consequently, the bonding layer is stiffer than the film (resp. more compliant) for $\b>0$ (resp. $\b<0$); the bonding layer is as stiff as film if $\b=0$.

\end{remark}

The study of equilibrium configurations corresponding to admissible global minimizers of the energy leads us to minimize $E(u)$ over the vector space of kinematically admissible displacements $\Cadz(\O)$.

Plugging the scalings above, the problem $\PeOe$ of finding the equilibrium configuration of the multilayer system depends implicitly on $\e$ via the assumed scaling laws, is defined on families of $\e$-dependent domains $(\Oe)_{\e>0}=(\Oef\cup\Oeb)_{\e>0}$, and reads:
\begin{equation}
\label{prb:min_3d_epsilon}
	 \PeOe: \quad \text{Find } \ue\in \Cadz(\Oe) \text{ minimizing } \tilde E_\e(u) \text{ among } v\in \Cadz(\Oe),
\end{equation}

Because the family of domains $(\Oe)_{\e >0}$ vary with $\e$ in $\PeOe$, we perform the classical anisotropic rescaling in order to state a new problem $ \PeO$, equivalent to $\PeOe$, in which the dependence upon $\e$ is explicit and is stated on a fixed domain $\O$.
Denoting by $x'=(x_1, x_2)\in \o$ and by $\tilde x'=(\tilde x_1,\tilde  x_2)$, the following anisotropic scalings:
\begin{equation}
\label{eqn:anis-scal-coord}
\pi^\e(x): \left\{
	\begin{aligned}
	&x=(x', x_3)\in \clo \Of \mapsto (\tilde x', \e\tilde  x_3) \in \clo\Oef, \\
		&x=(x',  x_3)\in \clo \Ob\mapsto (\tilde x', \e^{\a+1}\tilde  x_3) \in \clo\Oeb , 
\end{aligned}\right.
\end{equation}
map the domains $\Oef$ and $\Oeb$ into $\Of=\o\times [0, \hf)$ and $ \Ob=\o\times (-\hb, 0)$.
As a consequence of the domain mapping, the components of the linearized strain tensor $e_{ij}(v)=e^\e_{ij}(v\circ \pi(x))$ scale as follows:
\begin{gather}
	\eab^\e(v) \mapsto \eab(v), 
	 	\quad \ett^\e(v) \mapsto \frac{1}{\e}\ett(v), \quad \eat^\e(v)\mapsto \frac{1}{2}\left(  \frac{1}{\e} \pt \va + \pa \vt\right) \qquad  \text{in }\Oef,
	\\
	\eab^\e(v) \mapsto  \eab(v), 
	 	\quad \ett^\e(v) \mapsto \frac{1}{\e^{\a-1} } \ett(v), \quad \eat^\e(v)\mapsto \frac{1}{2}\left(  \frac{1}{\e^{\a-1}} \pt \va + \pa \vt\right) \qquad  \text{in }\Oeb. 
\end{gather}
Finally, the space of kinematically admissible displacements reads \[
	\Cadz(\O):=\left\{ v_i\in \H^1(\O),\; v_i=0\; \aee \onn  \o\times\{- \hb\} \right\}.
\]
It is easy to verify that the asymptotic minimization problem $\min_{u\in \Cadz(\O)} \hat E_\e(u) $ where  $\hat E_\e(u)=\frac{1}{\e}{\tilde E(u\circ \pi^\e(x))}$
yields the trivial convergence result $\ua = \lim_{\e\to0} \uea= 0$. 
This is to say that the in-plane components of the (weak limit) displacement are smaller than order zero in $\e$.
After having established this result, the analysis should be restarted anew to determine the convergence properties of the higher order terms.
Here, we skip that preliminary step and directly investigate the asymptotic behavior of the next order in-plane displacements, that is to say of fields $\tilde u^\e$ that admit the following scaling:
\begin{equation}
\label{eqn:resc_displ}
	\tilde u^\e = (\e \uea, \uet)\in \Cadz(\O).
\end{equation}
\begin{remark}
This result strongly depends upon the assumed scaling of external loads. Clearly, different choices rather than \eqref{eqn:scaling_law_loads} may lead to different scalings of the principal order of displacements, and possibly different limit models.
\end{remark}

Finally, dropping the tilde for the sake of simplicity, the parametric, asymptotic elasticity problem, stated on the fixed domain $\O$, using the scaling~\eqref{eqn:resc_displ} and in the regime of Hypothesis~\ref{hyp:elast_scaling_law}, reads:
\begin{equation}
	\PeO: \quad \text{Find } \ue\in \Cadz(\O) \text{ minimizing } E_\e(v) \text{ among } v \in \Cadz(\O),
\end{equation}
where, upon introducing the non-dimensional parameters
\begin{equation}
\label{eqn:nondim_par}
		\g:=\frac{\a+\b}{2}, \quad \d:=\frac{\b-\a}{2}-1,  \qquad\gamma, \delta\in \R,
\end{equation}
the scaled energy $E_\e(u)=\frac{1}{\e^3}{\tilde E(u\circ \pi^\e(x))}$ takes the following form:
\begin{multline}
\label{eqn:scaled_en}
	E_ \e(u)=\frac{1}{2}
	\int_\Of \left\{ \laf \left|\frac{\ett(u)}{\e^2}+\eaa(u)\right|^2
	+ \frac{2\muf }{\e^2} \left| \pt \ua + \pa \ut\right|^2 
	+ 2\muf \left( \left|\eab(u)\right|^2 +\left|\frac{\ett(u)}{\e^2}\right|^2  \right)  \right\} dx  \\
	+\frac{1}{2}\int_\Ob \left\{ \lab \left|\e^{\delta-1} \ett(u) +\e^{\gamma}\eaa(u)\right|^2  +
	 2\mub \left| \e^{\delta}\pt \ua + \e^{\gamma-1}\pa \ut\right|^2 
	+ 2\mub \left( \left|\e^{\gamma}\eab(u)\right|^2 +\left|\e^{\delta-1} \ett(u) \right|^2 \right)  \right\}dx\\
	-\int_\Of \left( 2 \muf \Att + \laf \Aaa \right) \frac{\ett(u)}{\e^2} dx
	-\int_\Of  \left\{ \laf \left(\Aaa + \Att \right) \ebb(u) + 2 \muf \Aab \eab(u) \right\}dx
	-\int_{\op} p u_3 dx' +F.
\end{multline}
In the last expression $F:=\frac{1}{2}\int_\Of (\elast_f)_{ijhk} \strload_{ij}:\strload_{hk}dx$ is the residual (constant) energy due to inelastic strains.
The non-dimensional parameters $\g$ and $\d$ represent the order of magnitude of the ratio between the membrane strain energy of the bonding layer and that of the film ($\g$), and the order of magnitude of the ratio between the transverse strain energy of the bonding layer and the membrane energy of the film ($\delta$).
They define a phase space, which we represent in Figure~\ref{fig:phase_diag}.

\begin{figure}[htbp]
 \centering
    \includegraphics[width=10cm]{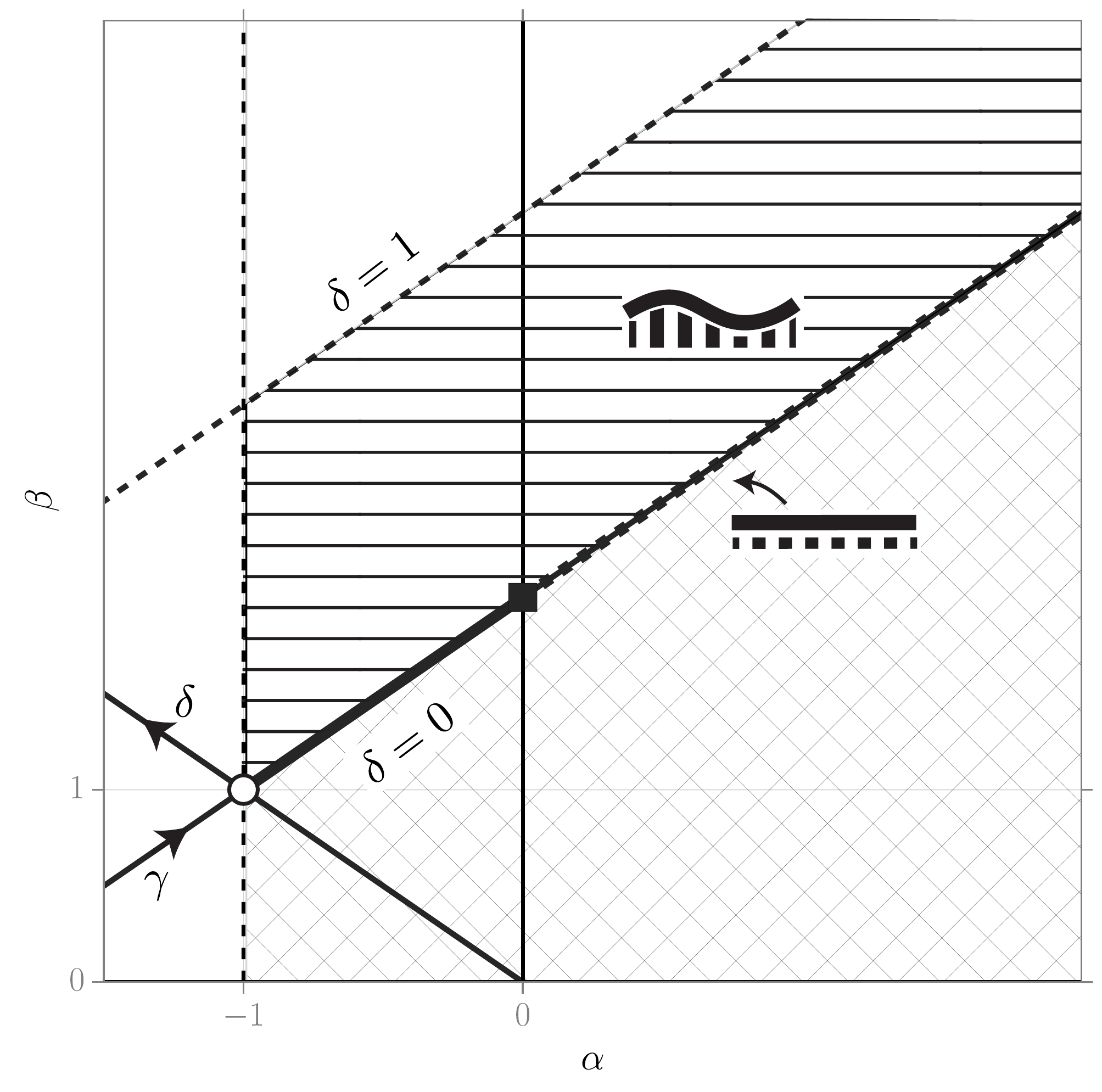}
 \caption{Phase diagram in the space ($\alpha$--$\beta$), where $\a$ and $\b$ define the scaling law of the relative thickness and stiffness of the layers, respectively. 
 Three-dimensional systems within the unshaded open region $\alpha<-1$ become more and more \emph{slender} as $\e\to0$.
  The square-hatched region represents systems behaving as ``rigid'' bodies, under the assumed scaling hypotheses on the loads.
 Along the open half line (displayed with a thick solid and dashed stroke) $(\delta, 0), \delta>0$ lay systems whose limit for vanishing thickness leads to a ``membrane over in-plane elastic foundation'' model, see Theorem~\ref{thm:in_of_plane_elast_found}.
 In particular, the solid segment $0<\gamma<1$ (resp. dashed open line $\gamma>1$) is related to systems in which bonding layer is thinner (resp. thicker) than the film, for $\g=1$ (black square) their thickness is of the same order of magnitude.
All systems within the red region $\gamma>0, 0<\delta\leq 1, \delta>\gamma$ behave, in the vanishing thickness limit, as ``plates over out-of-plane elastic foundation'', see Theorem~\ref{thm:in_of_plane_elast_found}.}
 \label{fig:phase_diag}
\end{figure}

The open plane $\g-\d<0$ corresponds to three-dimensional systems that become more and more slender as $\e\to0$.
Their asymptotic study conducts to establishing reduced, one-dimensional (beam-like) theories and falls outside of the scope of the present study.
The locus $\g-\d=0$ identifies the systems that stay three dimensional, as $\e\to0$, because the thickness of the bonding layer is always of order one (recall that $\clo{\Oeb} = \clo{\o}\times[-\e^{\a+1}\hb, 0]$ becomes independent of $\e$ for $\g-\d=0$).
In order to explore reduced, two-dimensional theories, we focus on the open half plane identified by:
\begin{equation}
\label{eqn:2d_regime_gammadelta_ineq}
 	\g-\d >0.
 \end{equation}

In what follows, we give a brief and non-technical account and mechanical interpretation of the dimension reduction results collected in Theorems~\ref{thm:in_of_plane_elast_found} and \ref{thm:out_of_plane_elast_found}.

For a given value of $\g$ and increasing values of $\d$ we explore systems in which  the order of magnitude of the energy associated to transverse variations of displacements in the bonding layer progressively increases relatively to the membrane energy of the film.
We hence encounter three distinct regions characterized by qualitatively different elastic couplings. Their boundaries are determined by the value of $\d$, as is $\d$ that determines the convergence properties of scaled displacements \eqref{eqn:resc_displ} at first order.
This argument will be made rigorous in Lemma~\ref{lem:poincare}.

For $\d<0$  the system is ``too stiff'' (relatively  to the selected intensity of loads) and both in-plane and transverse components of displacement vanish in the limit; that is, their order of magnitude is smaller than order zero in $\e$.

For $\d=0$, the shear energy of the bonding layer is of the same order of magnitude as the membrane energy of the film.
Consequently, elastic coupling intervenes between these two terms resulting in that the first order in-plane components of the limit displacements are of order zero.
Moreover, the transverse stretch energy of the bonding layer is singular and its membrane energy is infinitesimal: the first vanishes and the latter is negligible as $\e\to0$; the bonding layer undergoes purely shear deformations. 
More specifically, the condition of continuity of displacement at the interface $\op$ and the boundary condition on $\om$, both fix the intensity of the shear in the bonding layer.
As a consequence, the transverse profile of equilibrium (optimal) displacements is linear and the shear energy term in the bonding layer contributes to the asymptotic limit energy as a ``linear, in-plane, elastic foundation''.
On the other hand, because transverse stretch is asymptotically vanishing, out-of-plane displacements are constant along the thickness of the multilayer and are determined by the boundary condition on $\om$.
Hence, although Kirchhoff-Love coupling --\ie shear-free-- between components of displacements is allowed in the film, bending effects do not emerge in the first order limit model.
 More precisely, we are able to prove the following theorem:
\begin{theorem}[Membrane over in-plane elastic foundation]
\label{thm:in_of_plane_elast_found}
Assume that Hypotheses~\ref{hyp:geom_elast_scaling_law} and~\ref{hyp:elast_scaling_law} hold and let $\ue$ be the solution of Problem $\PeO$ for $\d=0$, then 
\begin{enumerate}[i)]
\item there exists a function $u \in \H^1(\OF; \R^3)$ such that $\ue \to u$  strongly in $\H^1(\OF; \R^3)$;
\item $u_3 \equiv 0$ and $\partial_3\ua \equiv 0$ in $\O$, so that $u$ can be identified with a function in $\H^1(\o,\R^2)$, which we still denote by $u$, and such that for all $v_\alpha \in  \H^1(\o,\R^2)$:
\begin{multline}
	 \io \left\{ \frac{2\laf \muf \hf}{\laf+2 \muf}\eaa(u) \ebb(v) 
    + 2\muf\hf \eab(u)\eab(v)
	+ \frac{2\mub}{\hb} \ua \va  \right\}\, dx'\\	
	= \io \left\{  \left(  c_1\bAaa +c_2\bAtt  \right) \ebb(v) +  c_3\bAab \eab(v) \right\}\, dx',
\end{multline}
where $\bar \strload_{ij} = \int_0^{\hf} \strload_{ij} dx_3$ are the averaged components of the inelastic strain over the film thickness, and coefficients $c_i$ are determined explicitly as functions of the material  parameters:
\begin{equation}
	\label{eqn:coeffs_in_plane}
c_1=
\frac{2 \laf  \muf }{\laf +2 \muf },
\quad c_2=
\frac{\laf ^2}{\laf +2 \muf },
\text{ and }\quad c_3=2 \muf. \end{equation}
\end{enumerate}

\end{theorem}
The last equation is interpreted as the variational formulation of the equilibrium problem of a \textbf{linear elastic membrane over a linear, in-plane, elastic foundation.}

In order to highlight the inherent \emph{size effect} emerging in the limit energy it suffices to normalize the domain $\o$ by rescaling the in-plane coordinates by a factor $L=\operatorname{diam}(\o)$. Hence, introducing the new spatial variable $y:=x'/L$ the equilibrium equations read:
\begin{multline}
\int_{\bar \o} \left\{  \eab(u)\eab(v) + 
\frac{\laf}{\laf+2 \muf}\eaa(u) \ebb(v) 
	+ \frac{ L^2 }{\ell_e^2 }\ua   \va  \right\}dy'\\
			= \int_{\bar \o} \left\{ \left(  \hat c_1 \bAaa +\hat c_2\bAtt  \right) \ebb(v) +   \hat c_3\bAab \eab(v) \right\}dy',
	\quad \forall \va \in \H^1(\o).
\end{multline}
where the internal elastic length scale of the membrane over in-plane foundation system is:
\begin{equation}
	\label{eqn:int_len_ip}
	\ell_e=\sqrt{\frac{\muf}{\mub}{\hf\hb}},
\end{equation}
and $\hat c_i = \frac{c_i}{2\muf \hf}$ and $\bar\o = \o/\operatorname{diam}(\o)$ is of unit diameter.
The presence of the elastic foundation, due to the non-homogeneity of the membrane and foundation energy terms, introduces a competition between the material, inherent, characteristic length scale $\ell_e$ and the diameter of the system $L$ and their ratio weights the elastic foundation term.

\bigskip

For $\d=1$, the transverse stretch energy of the bonding layer is of the same order as the membrane energy of the film and both shear and membrane energy of the bonding layer are infinitesimal.
The bonding layer can no longer store elastic energy by the means of shear deformations and in-plane displacements can undergo ``large'' transverse variations.
This mechanical behavior is interpreted as that of a layer allowed to ``slide'' on the substrate, still satisfying continuity of transverse displacements at the interface $\om$.
The loss of control (of the norm) of in-plane displacements within the bonding layer is due to the positive value of $\d$.
This requires enlarging the space of kinematically admissible displacements by relaxing the Dirichlet boundary condition on in-plane components of displacement on $\om$. This allows us to use a Korn-type inequality to infer their convergence properties.
Conversely, transverse displacements stay uniformly bounded within the entire system, the deformation mode of the bonding layer is a pure transverse stretch. 
In this regime, the value of the transverse strain is fixed by the mismatch between the film's and substrate's displacement, analogously to the shear term in the case of the in-plane elastic foundation.
Finally, from the optimality conditions (equilibrium equations in the bonding layer) follows that the profile of transverse displacements is linear and, owing to the continuity condition on $\oz$, they are coupled to displacement of the film.
The  latter undergoes shear-free (\ie Kirchhoff-Love) deformations and is subject to both inelastic strains and the transverse force.
This regime shows a stronger coupling between in-plane and transverse displacements of the two layers.
The associated limit model is that of a linear plate over a transverse, linear, elastic foundation.  
The qualitative behavior of system laying in the open region $\g,\d\in (\d, \infty)\times(0,1)$ is analogous to the limit case $\d=1$, although the order of magnitude of transverse displacements in the bonding layer differs by a factor $\e^{1-\d}$. 
More precisely, we are able to prove the following theorem:

\begin{theorem}[Plate over linear transverse elastic foundation]
\label{thm:out_of_plane_elast_found}
Assume that hypotheses~\ref{hyp:elast_scaling_law} and~\ref{hyp:geom_elast_scaling_law} hold and let $\ue$ denote the solution of Problem $\PeO$ for $0<\d\leq 1$, then:

\begin{enumerate}[i)]
\item the principal order of the displacement admits the scaling $\ue=(\e u_\alpha(\e), \e^{1-\d}u_3(\e))$;
\item there exists a function $u\in \CadrelKL(\OF)$ such that $\ue\to u$ converges strongly in $H^1(\Of)$;
\item the limit displacement $u$ belongs to the space $\CadrelKL(\O)$ and is a solution of the three-dimensional variational problem:\\
Find $u \in \CadrelKL(\O)$ such that:
\begin{multline}
	\label{eqn:lim_2d_plate_trsv_found}
	\iof \frac{2 \laf \muf}{\laf + 2 \muf} \eaa(u)\ebb(v) 
    + 2\muf \eab(u)\eab(v)\, dx
	+ \iob \frac{4 \mub (\lab +\mub)}{\lab + 2 \mub}\ett(u)\ett(v)  \, dx\\
	= \iof \left(c_4 \Aaa + c_5 \Att  \right)\ebb(v)  + c_6 \Aab \eab(v)\, dx + \int_{\op} p v_3\, dx',
\end{multline}
for all $v \in \CadrelKL(\O)$. Here, the in-plane displacement field $\ua$ is defined up to an infinitesimal rigid motion and the $c_i$'s are given by:
\begin{equation}
c_1=
\frac{2 \laf  \muf }{\laf +2 \muf },
\, c_2=
\frac{\laf ^2}{\laf +2 \muf },
\text{ and } c_3=2 \muf.
\label{eqn:coeffs_out_of_plane} 
\end{equation}
\item  
There exist two functions $\zeta_\alpha \in \H^1(\o)\cap \mathcal R(\Of)^\perp$ and $\zeta_3 \in \H^2(\o)$ such that the limit displacement field can be written under the following form:
\[
	\ua = \begin{cases}
		\zeta_\alpha(x'), &\text{ in }\OF\\
		\zeta_ \alpha(x') +(x_3+\hb)\pa \zeta_3(x'), &\text{ in }\OB,
	\end{cases} \qquad \text{and} \quad \ut=\zeta_3(x')\text{ in }\O,
\] 
and for all $\eta_\alpha \in \H^1(\o)\cap \mathcal R(\Of)^\perp,\eta_3 \in \H^2(\o)$ satisfies:
\begin{equation}
\label{eqn:limit_2d_plate}
	\begin{aligned}
		&\io 	\left\{ \frac{2 \laf \muf}{\laf + 2 \muf} 
 	 		\eaa(\eta)\ebb(\zeta) 
 		 		 	    + 2\muf  \eab(\eta)\eab(\zeta) \right\} dx'
           			= \io
	 \left(	c_1 \bAaa + c_2 \bAtt  \right) \ebb(\zeta)+  c_3 \bAab \eab(\zeta) dx',\\
 	&\io \left\{ \frac{ \laf \muf}{3(\laf + 2 \muf)} 
 	 		 		(\paa \eta_3\pbb\zeta_3)  
 		 	    +\frac{\muf}{3} 
        \pab \eta_3\pab\zeta_3
          	+ \frac{4 \mub (\lab +\mub)}{\lab + 2 \mub}\eta_3\zeta_3   \right\} dx'
		=
	 				 \int_{\o} p \zeta_3dx'.
	\end{aligned}
\end{equation}

\end{enumerate}
\end{theorem}
Equation~\eqref{eqn:lim_2d_plate_trsv_found} is interpreted as the variational formulation of the three-dimensional equilibrium problem of a \textbf{linear elastic plate over a linear, transverse, elastic foundation}, whereas Equations~\eqref{eqn:limit_2d_plate} are equivalent coupled, two-dimensional, flexural and membrane equations of a plate over a linear, transverse, elastic foundation in which components $\eta_\a$ and $\eta_3$ are respectively the in-plane and transverse components of the displacement of the middle surface of the film $\o\times\{\hf/2\}$.
This latter model is, strictly speaking, the two-dimensional extension of the Winkler model presented in the introduction.
Note that the solution of the in-plane problem above is unique only up to an infinitesimal rigid movement. This is a consequence of the loss of the Dirichlet boundary condition on for in-plane displacements in the limit problem. 
In addition, no further compatibility conditions are required on the external load, since it exerts zero work on infinitesimal in-plane rigid displacements.  
Similarly to the in-plane problem, the non-dimensional formulation of the equilibrium problems highlights the emergence of an internal, material length scale. 
Introducing the new spatial variable $y':=x'/L$ where $L=\operatorname{diam}(\o)$, the equilibrium equations read:
\begin{equation}
\label{eqn:limit_2d_plate_nd}
	\begin{aligned}
		&\int_{\bar \o} 	\left\{   \eab(\eta)\eab(\zeta) + \frac{\laf }{\laf + 2 \muf} 
 		\eaa(\eta)\ebb(\zeta) 
     \right\} dx'
	= \int_{\bar \o}
	 \left(	\hat c_1 \bAaa + \hat c_2 \bAtt  \right) \ebb(\zeta)  \hat c_3 \bAab \eab(\zeta) dy',\\
 	&\int_{\bar \o} \left\{ 
 	\pab \eta_3\pab\zeta_3 + 
 	\frac{ \laf }{\laf + 2 \muf} 
 		\paa \eta_3\pbb\zeta_3  
   	+ \frac{L^2}{\tilde\ell_e^2}\eta_3\zeta_3   \right\} dx'
	=
	 \int_{\bar \o} \hat p \zeta_3dy', \qquad \forall \zeta_\a \in H^1(\o), \zeta_3\in H^2(\o),
	\end{aligned}
\end{equation}
where the internal elastic length scale of the plate over transverse foundation system is:
\begin{equation}
	\label{eqn:int_length_oop}
	\tilde \ell_e = \sqrt{\frac{\muf(\lab + 2 \mub)}{12 \mub (\lab +\mub)}\hf \hb},
\end{equation}
$\hat p=\frac{p}{\muf \hf/3}$, and $c_i, \bar \o$ are the same as the definitions above.

\bigskip

The next section is devoted to the proof of the theorems.

\makeatletter{}\section{Proof of the dimension reduction theorems} \label{sec:proofs_of_the_theorems}
\subsection{Preliminary results} \label{sub:preliminary_results}

It is useful to introduce the notion of \emph{scaled strains}.
In the film, to an admissible field $v\in \H^1(\OF;\R^3)$ we associate the sequence of $\e$-indexed tensors $\kev\in L^2(\OF; \MtwoSym)$
whose components are defined by the following relations:
 \begin{equation}
	\label{eqn:resc_str_film}
 	\kttv=\frac{\ett(v)}{\e^2} ,
 	 \ktav=\frac{\eat(v)}{\e}, 
 	 \text{ and }   \kabv=\eab(v).
 \end{equation}

In the bonding layer, to an admissible field $v\in \left\{\hat v_i\in\H^1(\OB), \hat v_i=0 \onn \om \right\}$ we associate the tensor $\kevbl \in L^2(\OB; \MtwoSym )$, whose components are defined by the following relations:

\begin{equation}
	\label{eqn:resc_str_bl}
	 	\kttblv=\e^{\delta-1}\ett(v) , 
 	\ktablv=\frac{1}{2}\left( \e^\delta\pt \va + \e^{\gamma-1} \pa v_3 \right) ,  \text{ and } 
 	\kabblv=\e^\gamma\eab(v).
\end{equation}
Rewriting the energy~\eqref{eqn:scaled_en} the definitions above, the rescaled energy $ E_\e(v)$ reads:
\begin{multline}
\label{eqn:scaled_en_scaled_str}
 E_ \e(v)=
	\frac{1}{2}\iof \laf |\kttv +\kaav|^2 + 2\muf |\ktav|^2
	+ 2\muf \left( |\kttv|^2+|\kabv|^2 \right) \, dx\\
	+ \frac{1}{2}\iob \lab |\kttblv +\kaablv|^2 + 2\mub |\ktablv|^2
	+ 2\mub \left( |\kttblv|^2+|\kabblv|^2 \right)\, dx\\
	- \iof \left( 2\muf \Att + \laf \Aaa \right)  \kttv
	+ \laf \left(\Aaa + \Att \right)\kbbv +  2\muf \Aab \kabv \, dx\\
	-\int_{\op} p v_3 \, dx'+\int_\O (\elast_f)_{ijhk} \strload_{ij}:\strload_{hk}\, dx.
\end{multline}

The solution of the convex minimization problem $\PeO$ is also the unique solution of the following weak form of the first order stability conditions:
\begin{equation}
\label{eqn:vf_resc}
	\mathcal P(\e;\O):  \text{Find } \ue\in \Cadz(\O) \st  E_\e'(\ue)(v)=0,\; \forall v\in \Cadz.
\end{equation}
Here, by $E_\e'(u)(v)$ we denote the Gateaux derivative of $E_\e$ in the direction $v$. For ease of reference, its expression reads:
\begin{equation}
\label{eqn:gat_deriv_en}
	\begin{split}
E_\e'(u)(v)	&=\int_\Of \elast_f \ke{u}:\kev dx+\int_\Ob \elast_b \kebl{u}:\kevbl dx- \iof \elast \Ae :\ke{v} dx- \int_{\op} p v_3 dx'
\\
&=\iof \left\{ \left( (\laf+2\muf) \ktt{u} + \laf \kaa{u} \right) \kttv  + 2\muf \kta{u}\ktav  \right\}dx \\
	&\qquad+ \iof \left\{\laf \left( \ktt{u}+\kaa{u} \right) \kbbv + 2\muf \kab{u}\kabv \right\}dx\\
	&\qquad+ \iob \left\{ \left( (\lab+2\mub) \kttbl{u} + \lab \kaabl{u} \right) \kttblv  + 2\mub \ktabl{u}\ktablv \right\}dx \\
	&\qquad+ \iob \left\{  \lab \left( \kttbl{u}+\kaabl{u} \right) \kbbblv+ 2\mub \kabbl{u}\kabblv \right\}dx \\
	&\qquad- \iof \left\{  \left( 2\muf \Att + \laf \Aaa \right)  \kttv+ \laf \left(\Aaa + \Att \right)\kbbv +  2\muf \Aab \kabv \right\}dx
	-\int_{\op} p v_3 dx'. \\
\end{split}
\end{equation}

We establish preliminary results of convergence of scaled strains, using standard arguments based on a-priori energy estimates exploiting first order stability conditions for the energy. To this end, we need three straightforward consequences of Poincar\'e's inequality: one along a vertical segment, one on the upper surface and one in the bulk, which we collect in the following Lemma.

\begin{lemma}[Poincar\'e-type inequalities] 
Let $u\in L^2(\o)\times H^1(-\hb, \hf)$ with $u(x', -\hb)=0, \aee\ x'\in \o$. Then there exist two constants $C_1$ depending only on $\O$ and $C_2$ depending only on $\hf$ and $\hb$ such that:
\label{lem:poincare}
\begin{align}
\quad\nl[u(x', \cdot)]{(-\hb, \hf)}&\leq C_1(\hb, \hf) \left(  \nl[\pt u(x', \cdot)]{(0, \hf)}+ \nl[\pt u(x', \cdot)]{(-\hb, 0)} \right) \quad\aee\ x'\in \o,\label{eq:poincare_segment}\\
\nl[u]{\op}&\leq C_2(\O) \left( \nl[\pt u]{\Of} + \nl[\pt u]{\Ob} \right),\label{eq:poincare_surface}\\
\nl[u]{\O} & \leq C_2(\O) \left( \nl[\pt u]{\Of} + \nl[\pt u]{\Ob} \right).\label{eq:poincare_bulk}
\end{align}
\end{lemma}
\begin{proof}
Let $u\in L^2(\o)\times H^1(-1, 1)$ be such that  $u(x', -\hb)=0$ for \emph{a.e.}  $x' \in \o$. Then
\begin{align*}
	|u(x', x_3)| = |u(x', x_3)-u(x', -\hb)| = & \left|\int_{-\hb}^{x_3}\pt u(x', s)ds\right|\\
	 & \leq  \int_{-\hb}^{\hf}\left|\pt u(x', s) \right|ds \\
	 	 & \leq \nl[\pt u]{L^1(-\hb, \hf)}\\
	 & \leq (\hf + \hb)^{1/2}\nl[\pt u]{(-\hb, \hf)}
	\end{align*}
Consequently, on segments $\{x'\}\times(-\hb, \hf)$:
\[
	\begin{aligned}
	\nl[u(x', \cdot)]{(-\hb, \hf)} & \leq \left( \int_{-\hb}^{\hf} (\hf+\hb) \nl[\pt u]{(-\hb, \hf)}^2\right)^{1/2}\\
	 & \leq (\hf+\hb) \nl[\pt u]{(-\hb, \hf)}  
	\end{aligned}, 
\]
which gives the first inequality. On the upper surface $\op$:
\[
	\begin{aligned}
	\nl[ u]{\op}&\leq \left( \int_{\op} (\hf+\hb)^{1/2} \nl[\pt u]{(-\hb, \hf)}^2\right)^{1/2}\\
	& \leq |\O| \nl[\pt u]{\O}  
	\end{aligned}, 
\]
gives the second inequality. Finally, in the bulk:
\[
	\begin{aligned}
	\nl[ u]{\O}&=\left( \int_{\O}|u|^2 dx \right)^{1/2} 
	\leq \left( \int_{\o} \int_{-\hb}^{\hf} (\hf+\hb)^{1/2} \nl[\pt u]{L^2(-\hb, \hb)}^2\right)^{1/2}\\
	& \leq |\O| \nl[\pt u]{\O}  
	\end{aligned}, 
\]
which completes the claim.
\end{proof}
\begin{remark}
The crucial element in the above Poincar\'e-type inequalities is the existence of a Dirichlet boundary condition at the lower interface.
This allows to derive bounds on the components of displacements by integration over the entire surface $\o$, of the estimates constructed along segments $\{x'\}\times(-\hb, \hf)$.
\end{remark}

\begin{lemma}[Uniform bounds on the scaled strains]
\label{lem:conv_scaled_strains}
Suppose that hypotheses~\ref{hyp:geom_elast_scaling_law} and~\ref{hyp:elast_scaling_law} apply, and that $\delta \le 1$.  Let $\ue$ be the solution of $\mathcal P(\e;\O)$. Then, there exist constants $C_1,C_2>0$ such that for sufficiently small $\e$,
\begin{align}
	\nl[\kue]{\Of} & \le C_1\label{eq:conv_scaled_strains_1},\\
	\nl[\kttblue]{\Ob} & \le C_2\label{eq:conv_scaled_strains_2}.
\end{align}

\end{lemma}
\begin{proof}

Recalling that  $\rho_\mu = \mu_b / \mu_f$ we have :
\begin{align}
	2\muf \left( \nl[\kue]{\Of}^2+\rho_\mu\nl[\kttblue]{\Ob}^2 \right) &=  2\muf  \nl[\kue]{\Of}^2+ 2\mu_b\nl[\kttblue]{\Ob}^2 \notag \\
	& \le 2\muf  \nl[\kue]{\Of}^2+ 2\mu_b\nl[\kblue]{\Ob}^2 \notag \\
	& \le  \iof \elast_f \kue:\kue dx+ \iob \elast_b \kblue:\kblue dx,\notag 
\end{align}
where we have used the fact that $2\mu a_{ij} a_{ij}\leq \elast a:a$, which holds when $\elast$ is a Hooke tensor, for all symmetric tensors $a$, see~\cite{Ciarlet1997}.

Plugging $v=\ue$ in~\eqref{eqn:vf_resc}, we get that 
$$
	\iof \elast_f \kue:\kue dx+ \iob \elast_b \kblue:\kblue dx = \iof \elast_f \Ae :\kue dx + \int_{\op} \pe \uet dx',
$$
so that there exists a constant $C$ such that 
$$	  \nl[\kue]{\Of}^2+\rho_\mu\nl[\kttblue]{\Ob}^2 \le C\left( \nl[\kue]{\Of}+ \nl[\uet]{\op} \right),
$$
and for another constant (still denoted by $C$),
\begin{equation*}
	  \nl[\kue]{\Of}^2+\nl[\kttblue]{\Ob}^2 \le C\left( \nl[\kue]{\Of}+ \nl[\uet]{\op} \right).
	\label{eqn:energy_est}
\end{equation*}
Using the identity $(a+b)^2 \le 2(a^2+b^2)$, we get that 
\begin{equation*}
	  \left(\nl[\kue]{\Of}+\nl[\kttblue]{\Ob}\right)^2 \le C\left( \nl[\kue]{\Of}+ \nl[\uet]{\op} \right),
\end{equation*}
which combined with~\eqref{eq:poincare_surface} gives that
\begin{equation*}
	  \left(\nl[\kue]{\Of}+\nl[\kttblue]{\Ob}\right)^2 \le C\left( (1+\e^2) \nl[\kue]{\Of}+\e^{1-\delta}\nl[\kttblue]{\Ob} \right).
\end{equation*}
Recalling finally that $\delta \le 1$, we obtain~\eqref{eq:conv_scaled_strains_1} and~\eqref{eq:conv_scaled_strains_2} for sufficiently small $\e$.
\end{proof}
We are now in a position to prove the main dimension reduction results.

\subsection{Proof of Theorem~\ref{thm:in_of_plane_elast_found}} \label{sub:in_plane_elastic_foundation}

\newenvironment{statement}{\begin{em}}{\end{em}}

For ease of read, the proof is split into several steps.
\begin{enumerate}[i)]
\item\begin{statement}
Convergence of strains.
\end{statement}
Plugging~\eqref{eqn:resc_str_film} and~\eqref{eqn:resc_str_bl} in~\eqref{eq:conv_scaled_strains_1} and~\eqref{eq:conv_scaled_strains_2}, we have that
\begin{equation}
\label{eqn:bound_strains_f_dzero}
\nl[\ette(\ue)]{\OF}\leq  C\e^2, \quad
\nl[\eate(\ue)]{\OF}\leq C\e \text{, and }  
\nl[\eabe(\ue)]{\OF}\leq C; 
\end{equation}
and in the bonding layer:
\begin{equation}
\label{eqn:bound_strains_bl_dzero}
\nl[\ette(\ue)]{\OB}\leq C \e, \quad
\nl[\pt \uea]{\OB} \leq C,\quad 
\e^{\g-1}\nl[\pa \uet]{\OB} \leq C  \quad\text{and}\quad
\e^{\gamma}\nl[\eab(\ue)]{\OB}\leq C.
\end{equation}
These uniform bounds imply that there exist functions $\limeab\in L^2(\OF)$ such that $\eabe \wto \limeab$ weakly in $L^2(\OF)$, that $\eite(\ue)\to 0$ strongly in $L^2(\OF)$ and in particular that $\nl[\pt \uea]{\Of}\leq C\e$. 
Moreover $\ette(\ue)\to 0$ strongly in $L^2(\OB)$.
 \item\begin{statement}Convergence of scaled displacements.
\end{statement}

Using Lemma~\ref{lem:poincare} (Equation~\eqref{eq:poincare_bulk}) combined with \eqref{eqn:bound_strains_f_dzero} and~\eqref{eqn:bound_strains_bl_dzero} we can write:
\begin{subequations}
\label{eqn:bound_displ_dzero}
\begin{gather}
\label{eqn:bound_displ_transv_dzero}
		\nl[\uet]{\O} \leq C \left( \nl[\ett(\ue)]{\Of} + \nl[\ett(\ue)]{\Ob} \right) \leq C (\e^2 + \e) \leq C\e.
				\\ 
\label{eqn:bound_displ_inplane_dzero}
		 \nl[\uea]{\O} \leq C \left( \nl[\pt \uea]{\Of} + \nl[\pt \uea]{\Ob} \right) \leq C (\e + 1)\leq C.
		 \end{gather}
\end{subequations}
In addition, recalling from \eqref{eqn:bound_strains_f_dzero} that all components of the strain are bounded within the film, we infer that a function $u\in H^1(\Of)$ exists such that 
\begin{equation}
	\ue\to u \text{ strongly in } L^2(\Of), \text{ and } \ue \wto u \text{ weakly in } H^1(\Of).
\end{equation}
Similarly, by the uniform boundedness of $\ue$ in $L^2(\Ob)$,
it follows that $u$ can be extended to a function in $L^2(\O)$ such that
\begin{equation}
	\ue \wto u \text{ weakly in }L^2(\Ob).
\end{equation}
For $\aee\ x'\in \o$, we define the field $\vexp(x_3) = \ue(x', x_3)$.
Then $\vexp(x_3)\in H^1(-\hb, \hf)$ and, from the convergences established for $\ue$, it follows that there exists a function $v\in H^1(-\hb, \hf)$  such that $\vexp \wto v$ weakly in $H^1(-\hb, \hf)$, for $\aee\ x'\in \o$. 

Finally, from the first and second estimate in Equation~\eqref{eqn:bound_strains_f_dzero}, follows that the limit $u$ is such that $\eit(u)=0$, \ie the limit displacement belongs to the Kirchhoff-Love subspace $\CadKL(\Of)$ of sufficiently smooth shear-free displacements in the film.
Moreover, since the limit $u$ is such that $\pt \ua = 0$
 the in-plane limit displacement $\ua$ is independent of the transverse coordinate, that is to say:
\begin{equation}
	\label{eqn:ua_indep_x3_dzero}
 	\uea \wto \ua \text{ weakly }\inn H^1(\Of),
\end{equation}
where $\ua$ is independent of $x_3$, and hence it can be identified with a function $\ua\in H^1(\o)$, which we shall denote by the same symbol.

\item \begin{statement} Optimality conditions of the scaled strains.
The components of the weak limits $\kappa_{ij}\in L^2(\OF)$ 
 of subsequences of $\kue$ satisfy:
\begin{equation}
	\label{eqn:opt_scaled_strains_dzero}
		\limkttf= -\frac{\laf}{\laf+2 \muf}\limkaaf+ \frac{ 2\muf}{\laf+2 \muf}\Att +  \frac{ \laf}{\laf+2 \muf}\Aaa, \, 
	\limktaf=0, \text{ and } 
	\limkabf=\eab(u).
\end{equation}
\end{statement}
As a consequence of the uniform boundedness of sequences $\ke \ue$ and $\kebl \ue$ in $L^2(\Of; \MtwoSym)$ and $L^2(\Ob; \MtwoSym)$ established in Lemma~\ref{lem:conv_scaled_strains}, it follows that there exist functions $\limkf\in L^2(\Of, \MtwoSym)$ and $\limkbl \in L^2(\Ob; \MtwoSym)$ such that:
\begin{equation}
	\label{eqn:conv_resc_strains_dzero}
		\ke \ue \wto \limkf \text{ weakly in }L^2(\Of, \MtwoSym), \text{ and } 
	\kebl \ue \wto \limkbl \text{ weakly in }L^2(\Ob, \MtwoSym).
\end{equation}
The first two relations in~\eqref{eqn:opt_scaled_strains_dzero} descend from optimality conditions for the rescaled strains.
Indeed,  
taking in the variational formulation of the equilibrium problem test fields $v$ such that $\va= 0$ in $\O$, $ \vt = 0$ in $\OB$ and $ \vt \in \H^1(\OF)$ with $\vt =0$ on $\oz$ and multiplying by $\e^2$, we get:
\begin{multline}
\iof \left( (\laf+2\muf) \kttue + \laf \kaaue \right) {\ett(v)}  dx 
	=\iof \left\{ \left( 2\muf \Att + \laf \Aaa \right)  \ett(v)  \right\}dx +\\
	\e\iof 2\muf \ktaue \pa \vt
	+\e^2 \int_{\op} p \test \vt\,dx'.
	\end{multline}
Owing to the convergences established above for $\ke \ue$, $\kebl \ue $, 
since $\pa \vt$ and $\vt$ are uniformly bounded, we can pass to the limit $\e\to 0$ and obtain:
\[
	\int_{\OF}\left( (\laf+2 \muf)\limkttf + \laf \limkaaf \right)  \ett(v)dx= \int_{\OF} \left(  2\muf\Att +  \laf\Aaa  \right) \ett(v)dx.
	\]
From the arbitrariness of $v$, using arguments of the calculus of variations, we localize and integrate by parts further enforcing the boundary condition on $\oz$.
The optimality conditions in the bulk and the associated natural boundary conditions for the limit rescaled transverse strain $\limkttf$ follow:
\begin{equation}
\label{eqn:optimality_ktt_dzero}
		\limkttf = -\frac{\laf}{\laf+2 \muf}\limkaaf+ \frac{ 2\muf}{\laf+2 \muf}\Att +  \frac{ \laf}{\laf+2 \muf}\Aaa \quad \inn \Of, \qquad \text{and} \quad {\partial_3}\limkttf = 0 \onn \op.
\end{equation}

Similarly, consider 
 test fields $ v \in \H^1(\Of)$ such that $ \vt = 0$ in $\O$, $ \va = 0$ in $\OB$ and $ \va \in H^1(\OF)$ with $ \va=0$ on $\oz$.
 Multiplying the first order optimality conditions by $\e$, they take the following form:
\begin{multline}
\iof   2\muf \ktaue  \pt \va 
    dx
	= \e \iof  \left\{  \laf \left( \kttue+\kaaue \right) \ebb(v) 
    + 2\muf \kabue\eab(v) \right\}dx+\\
    \e\iof \left\{ \laf \left(\Aaa + \Att \right)\ebb(v) +  2\muf \Aab \eab(v) \right\}dx.
		 \end{multline}
The left-hand side converges to $\iof   2\muf \limktaf  \pt \va$ as $\e \to 0$, whereas  
the right-hand side converges to $ 0 $, since $ \eab(v)$ is bounded.
 We pass to the limit for $\e\to0$ and obtain:
\[
\iof  2\muf \limktaf  \pt \va =0.
\]
By integration by parts and enforcing boundary conditions we deduce that $\limktaf=0 \inn \Ob$, giving the second equation in~\eqref{eqn:optimality_ktt_dzero}.
Finally, by the definitions of rescaled strains~\eqref{eqn:resc_str_film} and the convergence of strains established in step i), we deduce that $\limkabf=\limeab$. But since $\ue \wto u $ in $H^1(\Of)$ implies the weak convergence of strains, in particular $\limeab = \eab(u)$, then
\[
	\limkabf=\eab(u),
	\]
which completes the claim.

\item \begin{statement}
Limit equilibrium equations
\end{statement}
Now, take test functions $v$ in the variational formulation of Equation~\eqref{eqn:vf_resc} 
  such that 
$\eit(v)=0$ in $\Of$ and $\ett(v)=0$ in $\Ob$, we get:
\begin{multline}
\iof \left\{ \laf \left( \kttue+\kaaue \right) \ebb(v) 
    + 2\muf \kabue\eab(v) \right\}dx
	+ \iob \left\{ 2\muf \ktablue\pt \va + \lab \left( \kttblue+\kaablue \right) \e\ebb(v) \right\}dx \\
	= \iof \left\{ \laf \left(\Aaa + \Att \right)\ebb(v) +  2\muf \Aab \eab(v) \right\}dx.
	\end{multline}
Since all sequences converge, we pass to the limit $\e\to 0$ using the first two optimality conditions in~\eqref{eqn:optimality_ktt_dzero} and obtain:
\begin{multline}
\label{eqn:limit_eq_equations_kappa_dzero}
\iof \left\{ \frac{2\muf\laf}{\laf+2 \muf}\limkaaf \ebb(v) 
    + 2\muf \limkabf\eab(v) \right\}dx
	+ \iob \left\{ 2\mub \pt \ua \pt \va \right\}dx\\
		= \iof \left\{ \left(c_1 \Aaa + c_2 \Att  \right)\ebb(v)  + c_3 \Aab \eab(v) \right\}dx
		 \end{multline}
where $c_1, c_2, c_3$ are the coefficients:
\[
c_1=
\frac{2 \laf  \muf }{\laf +2 \muf },
\quad c_2=
\frac{\laf ^2}{\laf +2 \muf },
\text{ and }\quad c_3=2 \muf. 
\]
Using the last relation in~\eqref{eqn:optimality_ktt_dzero} we obtain the variational formulation of the three-dimensional elastic equilibrium problem for the limit displacement $u$, reading:
\begin{multline}
\label{eqn:3d_limit_membrane_energy_unit_thick}
\iof   \left\{ \frac{2\laf \muf}{\laf+2 \muf}\eaa(u) \ebb(v) 
    + 2\muf \eab(u)\eab(v) \right\} dx
	+ \iob 2\mub \pt \ua  \pt \va dx\\
		= \iof  \left\{ \left(  c_1\Aaa +c_2\Att  \right) \ebb(v) +  c_3\Aab \eab(v) \right\}dx\\ \qquad \forall v \in \left\{ \test v_i\in \H^1(\O), \eit(\test v)=0 \inn \Of, \ett(\test v)=0 \inn \Ob \right\}.
\end{multline}

\item \begin{statement}
Two-dimensional problem.
\end{statement}

Owing to~\eqref{eqn:ua_indep_x3_dzero}, the in-plane limit displacement in the film is independent of the transverse coordinate;
let us hence consider test fields of the form:
\begin{equation}
v_ \alpha(x', x_3)=
 	\begin{cases}
 		\dfrac{(x_3+\hb)}{\hb} v_\alpha(x'), &\text{ in }\Ob\\
 		v_\alpha(x'), &\text{ in }\Of\\
 	\end{cases}, \text{ where } \va \in \H^1(\o).
 \end{equation}
 They provide pure shear and shear-free deformations in the bonding layer and film, respectively.
 For such test fields equilibrium equations read:
\begin{multline}
\io \left\{ \int_{0}^{\hf} \left( \frac{2\laf \muf}{\laf+2 \muf}\eaa(u) \ebb(v) 
    + 2\muf \eab(u)\eab(v) \right)dx_3 
	+ \int_{-\hb}^{0}  2\mub \pt \ua  \pt \va dx_3\right\}dx'\\
		= \io \left\{ \int_{0}^{\hf} \left( \left(  c_1\Aaa +c_2\Att  \right) \ebb(v) +  c_3\Aab \eab(v)\right)dx_3  \right\}dx'.
			\label{eqn:limit_eq_dzero_opt_prof}
\end{multline}
Recalling that $\ua$ is independent of the transverse coordinate in the film, and that for any admissible displacement $v\in \Cadz(\O)$ the following holds:
\[
	\int_{-\hb}^0 \pt  v(x', x_3)dx_3= v(x', 0)- v(x', -\hb)=v(x', 0), \quad\aee\ x'\in \o, 
\]
we integrate~\eqref{eqn:limit_eq_dzero_opt_prof} along the thickness and obtain:
\begin{multline*}
\io  \left\{ \frac{2\laf \muf \hf}{\laf+2 \muf}\eaa(u) \ebb(v) 
    + 2\muf \hf\eab(u)\eab(v)
	+ \frac{2\mub }{\hb} \ua(x', 0)   \va \right\}dx' \\
		= \io \left\{ \vphantom{\frac{\lambda}{\mu}} \hf\left(  c_1 \bAaa +c_2\bAtt  \right) \ebb(v) +  \hf c_3\bAab \eab(v) \right\}dx',
	\quad \forall \va \in \H^1(\o),
\end{multline*}
where overline   denote  averaging over the thickness: $\bar \Phi_{ij}:=\frac{1}{\hf}\int_0^{\hf} \Phi_{ij} dx_3$.
The last equation is the limit, two-dimensional, equilibrium problem for a linear elastic membrane on a linear, in-plane, elastic foundation and concludes the proof of item ii) in Theorem~\eqref{thm:in_of_plane_elast_found}.

\item \begin{statement}
Strong convergence in $\H^1(\OF)$
\end{statement} 

In order to prove the strong convergence of $\ue$ in $\H^1(\Of)$ it suffices to prove that $\nl[\eabe(\ue)-\eab(u)]{\Of}\to 0$ as $\e\to0$, as the strong convergence in  $L^2(\Of)$ of the components $\eite(\ue)$ has been already shown in step iii) of the proof.
Exploiting the convexity of the elastic energy, we can write:
\[
	\begin{split}
	2\muf\nl[\eabe(\ue)-\eab(u)]{\Of}&\leq 2\muf\nl[\kabue-\limkabf]{\O}\\ &\leq \iof \elast_f (\kue-\limkf):(\kue-\limkf)dx + \iob \elast_b (\kblue-\limkbl):(\kblue-\limkbl)dx \\
	&= \iof \elast_f \limkf:(\limkf - 2\kue)dx + \iob \elast_b \limkbl:(\limkbl- 2\kblue)dx \\
	& \qquad+ \iof \elast_f \kue:\kue dx+ \iob \elast_b \kblue:\kblue dx\\
	&= \iof \elast_f \limkf:(\limkf - 2\kue)dx + \iob \elast_b \limkbl:(\limkbl- 2\kblue)dx + \mathcal L(\ue).
\end{split}
\]
where the first inequality holds from the definitions of rescaled strains, and the last equality holds by virtue of the equilibrium equations (it suffices to take the admissible $\ue$ as test field in Equation~\eqref{eqn:vf_resc}).

By the convergences established for $\kue, \kblue$ and $\ue$, we can pass to the limit and get:
\[
	\lim_{\e\to0}\left( 2\muf\nl[\eabe(\ue)-\eab(u)]{\Of} \right) \leq \mathcal L(u) - \iof \elast_f \limkf:\limkf \, dx - \iob \elast_b \limkbl:\limkbl\, dx=0
\]
where the last equality gives the desired result and holds by virtue of the three-dimensional variational formulation of the limit equilibrium equations \eqref{eqn:limit_eq_equations_kappa_dzero}.
This concludes the proof of Theorem~\ref{thm:in_of_plane_elast_found}.
\end{enumerate}
\qed

\subsection{Proof of Theorem~\ref{thm:out_of_plane_elast_found}} \label{sec:out_of_plane_elastic_foundation}
For positive values of $\d$, elastic coupling intervenes between the transverse strain energy of the bonding layer and the membrane energy of the film, responsible of the asymptotic emergence of a reduced dimension model of a plate over an ``out-of-plane'' elastic foundation.

For ease of read, we first show the result for the case $\d=1$, splitting the proof into several steps.
\begin{enumerate}[i)]
\item \begin{statement}
Convergence of strains.
\end{statement}
Using the definitions of rescaled strains (Equations~\eqref{eqn:resc_str_film} and \eqref{eqn:resc_str_bl}), from the boundedness of sequences $\kue$ and $\kblue$ Lemma~\eqref{lem:conv_scaled_strains}, it follows that there exist constants $C>0$ such that, in the film:
\begin{equation}
\label{eqn:bound_strains_film_done}
\nl[\ette(\ue)]{\Of}\leq  C\e^2, \quad
\nl[\eate(\ue)]{\Of}\leq C\e , \quad \text {and} \quad  
\nl[\eabe(\ue)]{\Of}\leq C,
\end{equation}
and in the bonding layer
\begin{equation}
\label{eqn:bound_strains_bl_done}
\nl[\ette(\ue)]{\Ob}\leq C, \quad
 \nl[\pt \uea]{\Ob} \leq C \e^{-\delta} \quad\text{and}\quad
\e^{\gamma}\nl[\eab(\ue)]{\Ob}\leq C.
\end{equation}
These bounds, in turn, imply that there exist functions $\limeab\in L^2(\OF)$ such that $\eabe(\ue) \wto \limeab$ weakly in $L^2(\OF)$, a function $\limett\in L^2(\OB)$  such that $\ette(\ue)\wto \limett$ weakly in $L^2(\OB)$, and that $e_{i3}^\e(\ue)\to 0$ strongly in $L^2(\OF)$.

\item \begin{statement}
Convergence of scaled displacements.
\end{statement}

Using Lemma~\ref{lem:poincare} (Equation~\eqref{eq:poincare_bulk}) combined with \eqref{eqn:bound_strains_film_done} and~\eqref{eqn:bound_strains_bl_done} we can write:

\begin{equation*}
	\label{eqn:bound_displ_transv_done}
		\nl[\uet]{\O} \leq C \left( \nl[\ett(\ue)]{\Of} + \nl[\ett(\ue)]{\Ob} \right) \leq C (\e^2 + 1)\end{equation*}
from which, combined with \eqref{eqn:bound_strains_film_done}, follows that there exists a function $u_3 \in H^1(\O)$ such that $\partial_3 u_3 = 0$ in $\Of$, and 
\begin{equation}
	\label{eqn:conv_transv_displ_done}
	\uet \wto u_3 \text{ weakly in } H^1(\O). 
\end{equation}

By virtue of Korn's inequality in the quotient space $\dot  H^1(\Of)$ (see \eg, \cite{Ciarlet2013}) there exists $C>0$ such that
\[
	\nl[ u^\e_\alpha]{\dot H^1(\Of)} \leq C \nl[\eabe(u^\e_\alpha)]{L^2(\Of)},
\]
from which, recalling from \eqref{eqn:bound_strains_film_done} and denoting by $\Pi(\cdot)$ the projection operator over the space of rigid motions $\mathcal R(\Of)$, we infer that
  $\nl[ u^\e_\alpha - \Pi\left(u^\e_\alpha\right)]{H^1(\Of)}$ is uniformly bounded and hence, by the Rellich-Kondrachov Theorem that there exists $\ua \in H^1(\Of) \cap \mathcal R(\Of)^\perp$  such that 
 \begin{equation}
 \label{eqn:conv_in_plane_displ_done}
 	u^\e_\a -\Pi(u^\e_\a) \wto \ua\text{ weakly in }H^1(\Of).
 \end{equation}
Using then the second identity in~\eqref{eqn:bound_strains_film_done}, we have that $e_{i3}(u)=0$ in $\Of$, \ie that it belongs to the subspace of Kirchhoff-Love displacements in the film:
\[
  (\ua, u_3)\in \CadKL(\OF):=\left\{ \dot\H^1(\OF)\cap \mathcal(\Of)^\perp \times H^1(\Of), \mathrm{e}_{i3}(v)=0 \inn \OF \right\}.
\]

\item \begin{statement}
Optimality conditions of the scaled strains.
The components $\limkijf\in L^2(\OF)$ of the weak limits 
 of subsequences of $\kue$, and the component $\limkaabl\in L^2(\OB)$ of the weak limit of subsequences of $\kblue$, satisfy the following relations:
\begin{equation}
	\label{eqn:opt_resc_str_film_done}
	\limkttf= -\frac{\laf}{\laf+2 \muf}\limkaaf+ \frac{ 2\muf}{\laf+2 \muf}\Att +  \frac{ \laf}{\laf+2 \muf}\Aaa, 
	\;\limktaf=0,\text{ and }\;
	\limkabf=\eab(u)\quad \inn \Of
\end{equation}
and
\begin{equation}
	\label{eqn:opt_resc_str_bl_done}
	\limkaabl = -\frac{\lab}{\lab + 2 \mub}\limkttbl,\, \inn \Ob.
\end{equation}
\end{statement}
As a consequence of the uniform boundedness of sequences $\ke \ue$ and $\kebl \ue$ in $L^2(\Of; \MtwoSym)$ and $L^2(\Ob; \MtwoSym)$ established in Lemma~\ref{lem:conv_scaled_strains}, it follows that there exist functions $\limkf\in L^2(\Of, \MtwoSym)$ and $\limkbl \in L^2(\Ob; \MtwoSym)$ such that:
\begin{equation}
	\label{eqn:conv_resc_strains_done}
		\ke \ue \wto \limkf \text{ weakly in }L^2(\Of, \MtwoSym), \text{ and } 
	\kebl \ue \wto \limkbl \text{ weakly in }L^2(\Ob, \MtwoSym).
\end{equation}
The relations~\eqref{eqn:opt_resc_str_film_done} are established analogously to the case $\delta=0$, (see step iii) of Theorem~\ref{thm:in_of_plane_elast_found}) and their derivation is not reported here for conciseness.

To establish the optimality conditions~\eqref{eqn:opt_resc_str_bl_done} in the bonding layer, we start from~\eqref{eqn:vf_resc}, using test functions such that $v=0$ in $\OF$, $\vt=0$ in $\OB$ and $\va \in H^1_0(-\hb, 0)$ is a function of $x_3$ alone. For all such functions, dividing the variational equation by $\e$ we get:
\[
	\int_{\Ob} 2\mub\pt \ktablue v_\a' dx_3  = 0,
\]
which in turn yields that $\pt \ktablue=0$ in $\Ob$, \ie that the scaled strain $\ktablue$ is a function of $x'$ alone in $\Ob$.

Choosing test fields in the  variational formulation~\eqref{eqn:vf_resc} such that $\vt=0 \inn \Of, \vt = 0 \inn \Ob$, and $ \va = h_\alpha(x')g_\alpha(x_3)$ in $\Ob$ (no implicit summation assumed), where $h_\alpha(x')\in \H^1(\o), g_\alpha(x_3)\in\H^1_0(-\hb, 0)$, we obtain: 
\[
\begin{multlined}
	\int_\o \left\{ \int_{-\hb}^0  2\mub \ktablue \e h_\a g_\a'\,dx_3 +  \int_{-\hb}^0 \left(\lab \kttblue+ \left( \lab\delta_{\a \b}+2\mub \right) \kabblue \right) \e^\g \pb \va h_\a \,dx_3\right\}dx' =0.
	\end{multlined}
\]

The first term vanishes after integration by parts, using the boundary conditions on $g_\a$ and the fact that $\ktablue h_\alpha$ is a function of $x'$ only.
Dividing by $\e^\g$, we are left with:
\[
 \int_{-\hb}^0 \left[  \int_\o  \left(\lab \kttblue+ \left( \lab\delta_{\a \b}+2\mub \right) \kabblue \right) \pa h_\alpha \,dx' \right]g_\alpha dx_3=0.
  \]
We can use a localization argument owing to the arbitrariness of $g_\alpha$; moreover, since sequences $\kttblue, \kabblue$ converge weakly in $L^2(\Of)$, we can pass to the limit for $\e\to0$ and get for $\aee\ x'\in \o$: 
\[
   \int_\o  \left(\lab \limkttbl+ \left( \lab\delta_{\a \b}+2\mub \right) \limkabbl \right) \pa \vb\,dx'=0.
  \]
After an additional integration by parts, we finally obtain the optimality conditions in the bulk as well as the associated natural boundary conditions, namely:
\[
	\pb\left(\lab \limkttbl + \left( \lab\delta_{\a \b}+2\mub \right) \limkabbl \right) =0 \inn \o, 
	\text{ and }  
	\left(\lab \limkttbl + \left( \lab\delta_{\a \b}+2\mub \right) \limkabbl \right)n_\alpha =0 \onn \partial \o,
\]
where $n_\alpha$ denotes the components of outer unit normal vector to $\partial\o$.
In particular, optimality in the bulk for the diagonal term yields the desired result.

\item \begin{statement}
Limit equilibrium equations.
\end{statement}
We now establish the limit variational equations satisfied by the weak limit $u$.
Considering test functions $v\in H^1(\O)$ such that $v_3=0$ on $\om$ and $\eit(v)=0$ in $\OF$ in the variational formulation of the equilibrium problem~\eqref{eqn:vf_resc}, we get:
\[
	\begin{multlined}
\iof \left\{ \laf \left( \kttue+\kaaue \right) \ebb(v)  
    + 2\muf \kabue\eab(v)\right\}dx
	+ \iob \left( (\laf+2 \muf) \kttblue+\laf \kaablue \right)\ett(v)\,dx\\ 
	+ \iob \left\{  2\muf \ktablue \left( \e\pt \va + \e^{\gamma-1}\pa \vt\right)+ \lab \left( \kttblue+\kaablue \right) \e^\gamma\ebb(v) + 2 \mu \kabblue \e^\gamma \eab(v) \right\} dx\\ 
	= \iof \left\{ \laf \left(\Aaa + \Att \right)\ebb(v) +  2\mu \Aab \eab(v) \right\}dx + \int_{\op} p v_3\,dx',
	\end{multlined}
\]
Using again Lemma~\ref{lem:conv_scaled_strains}, and remarking that since $\gamma-1>0$ then $\e\pt \va$, $\e^{\gamma-1}\pa \vt$, and $\e^{\gamma}\eab(v)$ vanish as $\e \to 0$, we pass to the limit $\e\to0$ and obtain:
\[
\begin{multlined}
\iof \left\{ 
\laf \left( \limkttf+\limkaaf \right) \ebb(v)+ 2\muf \limkabf\eab(v)
 \right\} dx
	+ \iob \left( (\lab+2 \mub) \limkttbl+\lab \limkaabl \right)\ett(v) dx \\
		= \iof \left\{ \laf \left(\Aaa + \Att \right)\ebb(v) +  2\muf \Aab \eab(v) \right\} dx 
	+ \int_{\op} p v_3 \,dx',
\end{multlined}
\]
for all $v\in H^1(\O; \R^3)$ such that $v_3=0$ on $\om$ and $\eit(v)=0$ in $\OF$.
By the definitions of rescaled strains (Equations~\eqref{eqn:resc_str_film} and \eqref{eqn:resc_str_bl}) and plugging optimality conditions~\eqref{eqn:opt_resc_str_film_done} and~\eqref{eqn:opt_resc_str_bl_done}, we get:
\begin{multline}
\label{eqn:vf_lim_3d_done}
\iof \frac{2 \laf \muf}{\laf + 2 \muf} \eaa(u)\ebb(v) 
    + 2\muf \eab(u)\eab(v) dx
	+ \iob \frac{4 \mub (\lab +\mub)}{\lab + 2 \mub}\ett(u)\ett(v) dx \\
		= \iof \left(c_1 \Aaa + c_2 \Att  \right)\ebb(v)  + c_3 \Aab \eab(v) dx+ \int_{\op} p v_3 dx',
	\end{multline}
where the $c_i$'s are coefficients that depend on the elastic material parameters:
\[
c_1=\frac{2 \muf\laf }{\laf+2 \muf}, \quad c_2=\frac{\laf^2 }{\laf+2 \muf}, \quad  c_3=2 \muf. 
\]
Note that they coincide with those of the limit problem in Theorem~\ref{thm:in_of_plane_elast_found} since they descend from the optimality conditions within the film~\eqref{eqn:opt_resc_str_film_done}, which are the same.

\item\begin{statement}
Two-dimensional problem.
\end{statement}

As shown in step i), the limit displacement displacement satisfies $\eit(u)=0$.
Integrating these relations yields that there exist two functions $\eta_3\in H^2(\o)$  and  $\eta_\alpha \in H^1(\o)$, respectively representing the components of the out-of-plane and in-plane displacement of the middle surface of the film layer $\o\times\{\hf/2\}$,  such that $u \in \CadKL(\Of)$  is of the form:
\[
	u_ \alpha=\eta_\alpha(x')-(x_3-\hf/2)\pa\eta_3(x'), \text{ and }  u_3=\eta_3(x'). 
	\]
For such functions the components of the linearized strain read:
\[
	\eab(u)=\eab(\eta)-(x_3+\hf/2)\pab \eta_3 \text{ and }	\ett(u)=\ett(\eta).
\]
Analogously, there exist functions $\zeta_3\in \H^2(\o)$ and $\zeta_ \alpha \in  \H^1(\o)$ such that any admissible test field $v\in \left\{ v_i\in H^1(\O), v_3=0 \onn\ \om, \eit(v)=0 \inn \Of \right\}$ can be written in the form:
\[
	v_3=\begin{cases}
		\zeta_3(x'), &\text{ in } \Of\\
		(x_3+\hb)\zeta_3(x'), &\text{ in } \Ob\\
	\end{cases}, 
	\text{ and }
	v_\alpha=\zeta_\alpha(x')-(x_3+\hf/2) \pa \zeta_3(x'), \text{ in } \Of. 
		\]
The three-dimensional variational equation~\eqref{eqn:vf_lim_3d_done} can be hence rewritten as:
\[
\begin{multlined}
 	\iof \frac{2 \laf \muf}{\laf + 2 \muf} 
 	\left( \eaa(\eta)\ebb(\zeta) +(\paa \eta_3\pbb\zeta_3)(x_3-\hf/2)^2  +(x_3-\hf/2)\left( \eaa(\eta)\pbb\zeta_3 +\paa \eta_3\ebb(\zeta)\right)  \right) dx \\
    + \iof2\muf \left( \eab(\eta)\eab(\zeta)+(\pab \eta_3\pab\zeta_3)(x_3-\hf/2)^2 +(x_3-\hf/2)\left( \eab(\eta)\pab\zeta_3 +\pab \eta_3\eab(\zeta) \right)\right) dx\\ 
	+ \iob \frac{4 \muf (\lab +\muf)}{\lab + 2 \muf}\ett(\eta)\ett(\zeta)   dx
		= \iof \left\{ \left(c_1 \Aaa + c_2 \Att  \right)\ebb(\zeta)  + c_3 \Aab \eab(\zeta) \right\} dx + \int_{\o} p \zeta_3 dx, 
\end{multlined}
\] 
for all functions $\zeta_\alpha \in  \H^1(\o)$ and $\zeta_3\in \H^2(\o)$.
The dependence on $x_3$ is now explicit; after integration along the thickness the linear cross terms vanish in the film, and we are left with the two-dimensional variational formulation of the equilibrium equations:

\[
\begin{multlined}
	 	\io \frac{2 \laf \muf}{\laf + 2 \muf} 
 	\left\{
 		\eaa(\eta)\ebb(\zeta) +1/6(\paa \eta_3\pbb\zeta_3)  
 		 	\right\} dx'
    + \io 2\muf \left\{ \eab(\eta)\eab(\zeta)+1/6(\pab \eta_3\pab\zeta_3)
          \right\}dx'\\ 
	+ \io \frac{4 \mub (\lab +\mub)}{\lab + 2 \mub}\eta_3\zeta_3 dx'
		= \io \left\{ \left(c_1 \Aaa + c_2 \Att  \right)\ebb(\zeta)  + c_3 \Aab \eab(\zeta) \right\}dx' + \int_{\o} p \zeta_3dx', 
\end{multlined}
\]
for all functions $\zeta_\alpha \in  \H^1(\o)$ and $\zeta_3\in \H^2(\o)$.
By taking $\zeta_\alpha=0$ (resp. $\zeta_3=0$) the previous equation is broken down into two, two-dimensional variational equilibrium equations: the flexural and membrane  equilibrium equations of a Kirchhoff-Love plate over a transverse linear, elastic foundation.
They read:
\[
\begin{multlined}
	 	\io \left\{ \frac{2 \laf \muf}{\laf + 2 \muf} 
 	 		\eaa(\eta)\ebb(\zeta) 
 		 		 	    + 2\muf  \eab(\eta)\eab(\zeta) \right\} dx'
           			= \io
	 \left\{ \vphantom{\frac{\lambda}{\mu}}\left(	c_1 \Aaa + c_2 \Att  \right) \ebb(\zeta)  c_3 \Aab \eab(\zeta) \right\}dx', \\
		 \forall \zeta_\alpha \in  \H^1(\o),
\end{multlined}	
\]

\[
\begin{multlined}
	 	\io \left\{ \frac{ \laf \muf}{3(\laf + 2 \muf)} 
 	 		 		(\paa \eta_3\pbb\zeta_3)  
 		 	    +\frac{\muf}{3} 
        \pab \eta_3\pab\zeta_3
          	+ \frac{4 \mub (\lab +\mub)}{\lab + 2 \mub}\eta_3\zeta_3 \right\} dx'  
		=
	 				 \int_{\o} p \zeta_3 \,dx',\quad
	 \forall \zeta_3 \in \H^2(\o).
\end{multlined}	
\]

  To complete the proof in the case $0<\d<1$, it is sufficient to rescale transverse displacements within the bonding layer by a factor $\e^{1-\d}$, that is considering displacements of the form:
 \[
 	 	(\e\uea, \e^{1-\d}\uet) \inn \Ob 
 \]
	instead of \eqref{eqn:resc_displ}.
Then the estimates on the scaled strains leading to Lemma~\ref{lem:conv_scaled_strains}, as well as the arguments that follow, hold verbatim.
 \item \begin{statement}
 Strong convergence in $H^1(\Of)$
 \end{statement}
 The strong convergence $(\uea - \Pi(\uea), \uet)\to (\ua, \ut)$ in $H^1(\Of)$ is proved analogously to the case $\d=1$ (see step vi) in the proof of Theorem~\ref{thm:in_of_plane_elast_found}) and is not repeated here for conciseness.
\end{enumerate}
\qed

\makeatletter{}
We have studied the asymptotic behavior of a non-homogeneous, linear, elastic bi-layer in scalar elasticity.
Whenever a three-dimensional layer is ``thin'', its thickness-to-diameter ratio $\e$, appears naturally as a small parameter in the variational formulation of the equilibrium; it determines a \emph{singular perturbation} on the underlying problem of elasticity. 

On assuming a general scaling law for the two parameters, namely thickness and stiffness ratios, upon which the variational formulation of the elasticity problem solely depends,
we have unveiled and characterized the asymptotic regimes arising in the limit $\e\to 0$, that is solving the singular perturbation.
The asymptotic regimes synthetically resumed in the two-dimensional phase diagram of Figure~\ref{fig:phase_diag}, depending on stiffness and thickness ratio or equivalently on the two non-dimensional parameters $\g, \d$ representing the ratio of membrane energies and shear to membrane energy of the two layers, respectively.

The asymptotic limit regimes can also be hierarchically characterized by the order of magnitude (with respect to $\e$) of the leading term of the limit displacement $u$, \ie the order of the first non-trivial term in a possible power expansion with respect to $\e$. 

We identify the regime of membranes over a medium unable to transfer vertical stresses; that of bars under shear with added superficial stiffness; membranes over a three-dimensional body; higher order linear membranes under shear; and linear membranes over linear elastic foundation.
The latter regime is of particular interest: the asymptotic analysis rigorously justifies the widely adopted linear elastic foundation model (\ie the ``Winkler foundation'' or ``shear lag'' model).
It is further established a \emph{class of equivalence} of three-dimensional elastic bi-layers having the same limit representation, giving insight into the nature and validity of the aforementioned reduced model. In addition, an explicit equation allows to identify the single parameter identifying the linear foundation model, as a function of the three-dimensional material and geometric parameters.

The asymptotic study is performed in the simplified setting of linearized scalar elasticity. 
Since we are mainly interested into the in-plane behavior, for the reasons illustrated in the body of the work, this setting is rich enough to unveil the basic elastic energy coupling mechanisms.
The same does not hold if we were interested in the out-of-plane behavior, \ie to study bending effects.
However, a similar asymptotic analysis could be carried with the same spirit, at the expense of a more involved analytic treatment, in the framework of linear, three-dimensional, vectorial, elasticity. 
In both cases, the assumption of linearized elasticity is delicate. In the context of genuinely nonlinear elasticity, indeed, the order of magnitude of the applied loads plays a crucial role in determining the limit asymptotic regimes unlike in the linear case, where it can be transparently rescaled.

The analysis presented here is an effort to show how reduced-dimension models, often regarded as constitutive, phenomenological models, can rigorously derive and be justified from genuine three-dimensional elasticity. 
This, not only provides a mathematically sound treatment, but also gives insight into the fundamental elastic mechanisms and the nature of their coupling, further supplying the range of validity of the reduced models and hence an essential indication for their practical application.

\section*{Acknowledgements}
	B. Bourdin's work was supported in part by the National Science Foundation under the grant DMS-1312739. A. A. León Baldelli's work was supported by the Center for Computation \& Technology at LSU.

% \printbibliography
\bibliographystyle{plain}
\bibliography{biblio}

\begin{thebibliography}{10}

\bibitem{Amabili1998}
M~Amabili, MP~Pa\"idoussis, and AA~Lakis.
\newblock {Vibrations of partially filled cylindrical tanks with
  ring-stiffeners and flexible bottom}.
\newblock {\em Journal of Sound and Vibration}, 213:259--299, 1998.

\bibitem{Audoly2008a}
Basile Audoly and Arezki Boudaoud.
\newblock {Buckling of a stiff film bound to a compliant substrate—Part I:
  Formulation, linear stability of cylindrical patterns, secondary
  bifurcations}.
\newblock {\em Journal of the Mechanics and Physics of Solids},
  56(7):2401--2421, July 2008.

\bibitem{Audoly2008b}
Basile Audoly and Arezki Boudaoud.
\newblock {Buckling of a stiff film bound to a compliant substrate—Part II: A
  Global Scenario for the Formation of the Herringbone pattern}.
\newblock {\em Journal of the Mechanics and Physics of Solids},
  56(7):2422--2443, July 2008.

\bibitem{Audoly2008c}
Basile Audoly and Arezki Boudaoud.
\newblock {Buckling of a stiff film bound to a compliant substrate—Part III:
  Herringbone solutions at large buckling parameter}.
\newblock {\em Journal of the Mechanics and Physics of Solids},
  56(7):2422--2443, July 2008.

\bibitem{Chen2003}
Xing-Chong Chen and Yuan-Ming Lai.
\newblock {Seismic response of bridge piers on elasto-plastic Winkler
  foundation allowed to uplift}.
\newblock {\em Journal of Sound and Vibration}, 266(5):957--965, October 2003.

\bibitem{Ciarlet1996c}
PG~Ciarlet and V\'{e}ronique Lods.
\newblock {Asymptotic analysis of linearly elastic shells. I. Justification of
  membrane shell equations}.
\newblock {\em Archive for rational mechanics and analysis}, 136(2):119--161,
  December 1996.

\bibitem{Ciarlet1997}
Philippe~G. Ciarlet.
\newblock {\em {Mathematical Elasticity Volume II: Theory of Plates}}.
\newblock North-Holland, Amsterdam, 1997.

\bibitem{Ciarlet2013}
Philippe~G Ciarlet.
\newblock {\em {Linear and Nonlinear Functional Analysis with Applications}}.
\newblock SIAM, 2013.

\bibitem{Ciarlet1979}
Philippe~G. Ciarlet and Philippe Destuynder.
\newblock {A Justification of a Nonlinear Model in Plate Theory}.
\newblock {\em Computer Methods in Applied Mechanics and Engineering}, 18,
  1979.

\bibitem{Cox1952a}
H~L Cox.
\newblock {The elasticity and strength of paper and other fibrous materials}.
\newblock {\em British Journal of Applied Physics}, 3(3):72--79, March 1952.

\bibitem{Filonenko-Borodich1940}
M.~M. Filonenko-Borodich.
\newblock {Some approximate theories of the elastic foundation}.
\newblock {\em Uchenyie Zapiski Moskovskogo Gosudarstvennogo Universiteta
  Mekhanica}, 46:3--18, 1940.

\bibitem{Friesecke2006}
Gero Friesecke, Richard~D James, and Stefan M\"{u}ller.
\newblock {A Hierarchy of Plate Models Derived from Nonlinear Elasticity by
  Gamma-Convergence}.
\newblock {\em Archive for Rational Mechanics and Analysis}, 180(2):183--236,
  January 2006.

\bibitem{Gerolymos2006}
Nikos Gerolymos and George Gazetas.
\newblock {Winkler model for lateral response of rigid caisson foundations in
  linear soil}.
\newblock {\em Soil Dynamics and Earthquake Engineering}, 26(5):347--361, May
  2006.

\bibitem{Geymonat1987a}
Giuseppe Geymonat, Fran\c{c}oise Krasucki, and Jean-Jacques Marigo.
\newblock {Sur la commutativit\'{e} des passages \`{a} la limite en th\'{e}orie
  asymptotique des poutres composites}.
\newblock {\em Comptes Rendus de l'Acad\'{e}mie des Sciences - Series II -
  M\'{e}canique des solides}, pages 225--228, 1987.

\bibitem{Hutchinson1990}
JW~Hutchinson and HM~Jensen.
\newblock {Models of fiber debonding and pullout in brittle composites with
  friction}.
\newblock {\em Mechanics of Materials}, 9:139--163, 1990.

\bibitem{Jiang2008}
Guoliang Jiang and Kara Peters.
\newblock {A shear-lag model for three-dimensional, unidirectional multilayered
  structures}.
\newblock {\em International Journal of Solids and Structures},
  45(14-15):4049--4067, July 2008.

\bibitem{Kleckner2004}
Nancy Kleckner, Denise Zickler, Gareth~H Jones, Job Dekker, Ruth Padmore, Jim
  Henle, and John Hutchinson.
\newblock {A mechanical basis for chromosome function.}
\newblock {\em Proceedings of the National Academy of Sciences of the United
  States of America}, 101(34):12592--7, August 2004.

\bibitem{LeonBaldelli2014}
A.~A. {Le\'{o}n Baldelli}, J.-F. Babadjian, B.~Bourdin, D.~Henao, and
  C.~Maurini.
\newblock {A variational model for fracture and debonding of thin films under
  in-plane loadings}.
\newblock {\em Journal of the Mechanics and Physics of Solids}, June 2014.

\bibitem{LeonBaldelli2012}
Andr\'{e}s~Alessandro {Le\'{o}n Baldelli}, Blaise Bourdin, Jean-Jacques Marigo,
  and Corrado Maurini.
\newblock {Fracture and debonding of a thin film on a stiff substrate:
  analytical and numerical solutions of a one-dimensional variational model}.
\newblock {\em Continuum Mechanics and Thermodynamics}, 25(2-4):243--268, May
  2013.

\bibitem{Lions1973}
Jacques-Louis Lions.
\newblock {\em {Perturbations Singulieres dans les Problemes aux Limites}}.
\newblock Springer-Verlag, Berlin, Heidelberg, New York, 1973.

\bibitem{Marigo2005a}
Jean-Jacques Marigo and Kamyar Madani.
\newblock {The influence of the type of loading on the asymptotic behavior of
  slender elastic rings}.
\newblock {\em Journal of Elasticity}, 75(2):91--124, January 2005.

\bibitem{Mesgarnejad}
Ataollah Mesgarnejad, Blaise Bourdin, and M.M. Khonsari.
\newblock {A variational approach to the fracture of brittle thin films subject
  to out-of-plane loading}.
\newblock {\em Journal of the Mechanics and Physics of Solids}, pages 1--32,
  May 2013.

\bibitem{Nairn1997}
JA~Nairn.
\newblock {On the use of shear-lag methods for analysis of stress transfer in
  unidirectional composites}.
\newblock {\em Mechanics of Materials}, 26:63--80, 1997.

\bibitem{Nairn2001}
J.a. Nairn and D.a. Mendels.
\newblock {On the use of planar shear-lag methods for stress-transfer analysis
  of multilayered composites}.
\newblock {\em Mechanics of Materials}, 33(6):335--362, June 2001.

\bibitem{Paliwal1996}
D.N. Paliwal, Rajesh~Kumar Pandey, and Triloki Nath.
\newblock {Free vibrations of circular cylindrical shell on Winkler and
  Pasternak foundations}.
\newblock {\em International Journal of Pressure Vessels and Piping},
  69(1):79--89, November 1996.

\bibitem{Pasternak1955}
P.~L. Pasternak.
\newblock {On a New Method of Analysis of an Elastic Foundation by Means of Two
  Foundation Constants}.
\newblock {\em Gosudarstvennoye Izdatel'stvo Literatury po Stroitel'stvu I
  Arkhitekture}, 1955.

\bibitem{Pradhan2009}
S.C. Pradhan and J.K. Phadikar.
\newblock {Nonlocal elasticity theory for vibration of nanoplates}.
\newblock {\em Journal of Sound and Vibration}, 325(1-2):206--223, August 2009.

\bibitem{Pradhan2011}
S.C. Pradhan and G.K. Reddy.
\newblock {Buckling analysis of single walled carbon nanotube on Winkler
  foundation using nonlocal elasticity theory and DTM}.
\newblock {\em Computational Materials Science}, 50(3):1052--1056, January
  2011.

\bibitem{Sica2013}
Stefania Sica, George Mylonakis, and Armando~Lucio Simonelli.
\newblock {Strain effects on kinematic pile bending in layered soil}.
\newblock {\em Soil Dynamics and Earthquake Engineering}, 49:231--242, June
  2013.

\bibitem{Winkler1867}
Emil Winkler.
\newblock {\em {Die Lehre von der Elasticit\"{a}t und Festigkeit}}.
\newblock 1867.

\bibitem{Xia2000}
Z~Cedric Xia and John~W Hutchinson.
\newblock {Crack patterns in thin films}.
\newblock {\em Journal of the Mechanics and Physics of Solids}, 48:1107--1131,
  2000.

\bibitem{Zafeirakos2013}
A~Zafeirakos, N~Gerolymos, and V~Drosos.
\newblock {Incremental dynamic analysis of caisson–pier interaction}.
\newblock {\em Soil Dynamics and Earthquake Engineering}, 48:71--88, May 2013.

\end{thebibliography}

\end{document}